\documentclass[12pt]{extarticle}
\usepackage[utf8]{inputenc}
\usepackage{latexsym}
\usepackage{soul}
\usepackage{amssymb,amsmath, amsthm}
\usepackage{pifont}
\usepackage{listings}
\usepackage{float} 
\usepackage{mathrsfs}
\usepackage[pdftex]{graphicx}
\graphicspath{{figs/}}
\usepackage{booktabs}
\usepackage{ragged2e}
\usepackage{tikz}
\usetikzlibrary{intersections, calc}
\usetikzlibrary{patterns}
\usetikzlibrary{3d}
\usepackage{changepage}
\usepackage{setspace}
\usepackage{geometry}  
\usepackage{hyperref}  
\usepackage[utf8]{inputenc}
\usepackage[affil-it]{authblk} 
\usepackage{etoolbox}
\usepackage{lmodern}
\usepackage[most]{tcolorbox}
\usepackage{pst-node}
\usepackage{tikz-cd} 
\usepackage{empheq}
\usepackage{cancel}
\usepackage{ dsfont }
\usepackage{enumitem}
\usepackage{chngcntr}
\usepackage{quiver} 

\allowdisplaybreaks

\usepackage{titling}

\usepackage{etoolbox}

\usepackage{natbib}
\newcommand{\cites}[1]{\citeauthor{#1}'s (\citeyear{#1})}

\theoremstyle{definition}
\newtheorem{example}{Example}
\newtheorem{remark}{Remark}

\theoremstyle{plain}
\newtheorem{definition}{Definition}

\newtheorem{proposition}{Proposition}
\newtheorem{fact}{Fact}
\newtheorem{axiom}{Axiom}
\newtheorem{theorem}{Theorem}

\newtheorem{lemma}{Lemma}

\usepackage{array} 
\newcolumntype{C}[1]{>{\centering\arraybackslash}p{#1}}

\usepackage[most]{tcolorbox}

\tcbset{
    breakable,
    enhanced,
    arc=0mm,
    outer arc=0mm,
    boxrule=0.3mm 
}

\makeatletter
\DeclareRobustCommand{\varamalg}{%
  \mathbin{\mathpalette\var@malg\perp}%
}
\newcommand{\succprec}{\mathrel{\mathpalette\succ@prec{\succ\prec}}}
\newcommand{\precsucc}{\mathrel{\mathpalette\succ@prec{\prec\succ}}}

\newcommand{\succ@prec}[2]{\succ@@prec#1#2}
\newcommand{\succ@@prec}[3]{%
  \vcenter{\m@th\offinterlineskip
    \sbox\z@{$#1#3$}%
    \hbox{$#1#2$}\kern-0.4\ht\z@\box\z@
  }%
}
\newcommand\var@malg[2]{%
  \rlap{$\m@th#1#2$}\mkern6mu{#1#2}%
}
\makeatother

\newcommand{\pr}[0]{\mathbb{P}}

\newcommand{\R}[0]{\mathbb{R}}

\newcommand{\1}[0]{\mathds{1}}

\newcommand{\E}[0]{\mathbb{E}}

\newcommand{\spiff}[0] {\hspace{10pt}  \iff \hspace{10pt}}

\DeclareMathOperator*{\argmax}{arg\,max}
\DeclareMathOperator*{\argmin}{arg\,min}

\DeclareMathOperator{\Bern}{Bern}

\hypersetup{hidelinks}

\title{Dynamically Consistent Statistical Decisions\thanks{We thank Isaiah Andrews, Ophir Averbuch,  Shengwu Li, Matthew Rabin, Jesse Shapiro, Tomasz Strzalecki, Vitalii Tubdenov, Davide Viviano, and workshop participants at Harvard University for their valuable feedback.}} 

\date{ July 11, 2026 }

\author{
Cheaheon  Lim\thanks{Department of Economics, Harvard University.}  \and Yechan Park\thanks{Department of Economics, Harvard University.} 
}

\begin{document}

\maketitle

\begin{center}
\textbf{Abstract}
\end{center}

\begin{adjustwidth}{1.5cm}{1.5cm}
\onehalfspacing

A large literature in econometrics proposes decision rules with optimality guarantees based on ex ante criteria, such as minimax regret.  We develop a framework for analyzing the dynamic consistency of such rules and show that, in many empirically relevant settings, the researcher may wish to deviate from the interim prescription of ex ante optimal rules  after observing the data realization. To address this problem, we propose and axiomatize two classes of optimality criteria that yield dynamically consistent decision rules. 
\end{adjustwidth}



\newpage
\onehalfspacing

\section{Introduction}

Statistical decision theory  à la  \cite{Wald_1950_SDT} provides a framework for studying the ex ante choice of data-contingent actions. In the context of treatment choice, for instance, the decision maker (DM) must commit to a decision rule that either approves or rejects the treatment as a function of the realized data \citep[e.g.,][]{Manski_2004_ECTA, Stoye_2009_JoE}. This ex ante perspective is most natural when the DM can commit ex ante to a decision rule before observing the realized data.  

In practice, however, DMs often revise their intended action after seeing the data. For instance, a regulatory agency with a publicly listed approval rule for new drugs may reject a drug after receiving its clinical trial data, even when the data satisfies all the conditions required for approval under the rule.   In the context of empirical research, a researcher may commit to using a nonparametric estimator in a pre-analysis plan, but instead opt for a more parametric specification after observing that the resulting estimate is too noisy. If these DMs had fully specified their data-contingent actions at the ex ante stage,  such behavior constitutes instances of \emph{dynamic inconsistency:} after observing the data, the DM  wishes to deviate from the decision rule chosen ex ante.


A prerequisite for discussing notions of dynamic consistency is an \emph{interim} optimality criterion, which describes the DM's preferences after the realization of data.\footnote{Our terminology distinguishes between the \emph{interim} stage, which occurs  after the data realization but  before the realization of the parameter space,  and the \emph{ex post} stage, which occurs after the realization of all uncertainty.}  These interim preferences determine whether the DM wishes to deviate from a given decision rule once the data has realized. In Bayesian analysis, the interim problem is well-defined by the DM's posterior beliefs after the realization of data. By contrast, in frequentist decision theory, the interim problem is less clearly defined. There are no beliefs to update in the Bayesian sense, and standard ex ante criteria, such as  minimax regret, do not admit a natural decomposition into interim preferences.


The first main contribution of this paper is to provide a formal analysis of the interim problem associated with frequentist ex ante optimality criteria. This fills an important gap in the literature, given the prominence of decision rules justified by such criteria.  Through a series of examples, we illustrate  that dynamic inconsistency is not merely a theoretical possibility. Across a variety of empirically relevant settings, many ``reasonable" researchers and policymakers may find it optimal, after observing the data, to deviate from the ex ante optimal rule.\footnote{Related concerns appear in the statistics literature (e.g., Chapter 1.6 of  \cite{Berger_1985_SDT}).
}  

The second main contribution of this paper is to propose and axiomatize a class of ex ante optimality criteria that are {dynamically consistent}, and therefore do not suffer from the aforementioned problems. These correspond to the ex ante preferences of a \emph{sophisticated} agent in the sense of \cite{Strotz_ReStud_1955}, who correctly anticipates the interim preferences she will have after the data realization.  Normatively, such criteria are appealing because they select decision rules that are optimal not only ex ante, but also after each possible data realization.\footnote{Computationally, this means  that optimal rules can be computed via backwards induction.} Hence, such decision rules are \emph{credible}: the DM  has no incentive to deviate from the  action prescribed by the ex ante rule after the data realization.   We show that our dynamically consistent optimality criteria nest, as special cases, the as-if optimization of \cite{Manski_2021_ECTA} and Gamma${}^*$-minimax criterion of \cite{Lim_2026_PIA}.

At a high level, the dynamically consistent optimality criteria we characterize can be described as semi-Bayesian. They resemble Bayesian criteria in that they evaluate decision rules by aggregating continuation payoffs across possible data realizations. They differ from ordinary Bayesian criteria, however, in the role assigned to prior beliefs. Here, the DM is assumed to have enough prior information to form an ex ante marginal distribution over the data, but the interim criterion used after observing the data need not be obtained by Bayesian updating of a prior over the state space. Instead,  the DM can use any statistical procedure (e.g., worst-case evaluation over a confidence region) after the data is realized.


This distinction is important because the underlying state space may  have \emph{``no concrete reality in terms of physical quantities that are easily accessible to intuition, and yet the phenomenon under study may be quite familiar to the investigator"} \citep[p.~95]{Berger_1985_SDT}. That is, requiring the DM to have well-defined prior beliefs over the state space, in which case the marginal equals the prior predictive, is a more restrictive assumption than assuming that she has a subjective marginal distribution.   For instance, a development economist evaluating a cash-transfer program may be unable to specify a  prior over the collection of all joint distributions of potential outcomes, but she may have  historically grounded beliefs about the distribution of study results likely to arise. 

Finally, we remark that our  framework can also describe a Bayesian DM with prior beliefs who nevertheless performs frequentist procedures using the data. This interpretation captures a familiar feature of empirical practice, where researchers typically adhere to the frequentist paradigm despite having substantive prior information about the phenomenon under study.

\subsection{Related Literature}

Since the early contributions of \cite{Manski_2004_ECTA} and \cite{Dehejia_JoE_2005}, a growing literature in econometrics applies statistical decision theory to economically meaningful problems, such as treatment choice and policy learning. In the frequentist setting, finite-sample minimax-regret optimal treatment rules have been derived across a variety of settings \citep[e.g.,][]{Stoye_2009_JoE, tetenov2012statistical, yata2021optimal, chen2025note, guggenberger2026minimax}. Many decision rules based on asymptotic minimax-regret guarantees have also been proposed \citep[e.g.,][]{hirano2009asymptotics, Kitagawa_Tetenov_2018_ECTA, athey2021policy, mbakop2021model, Chernozhukov_et_al_2025, Sun_2026_EWM}. Recent work by \cite{fernandez2024robust}, \cite{Giacomini_Kitagawa_Read_2025}, and \cite{Christensen_et_al_Restud_2026} studies related problems from the robust Bayesian perspective.

The aforementioned papers focus primarily on ex ante optimality. Our contribution is to ask whether rules justified by such ex ante criteria remain optimal after the DM observes the data.\footnote{This question also connects to recent work on pre-analysis plans, which studies the value of ex ante commitment and the consequences of post-data deviations from preregistered analyses \citep[e.g.,][]{olken2015promises, kasy2024optimalpreanalysisplansstatistical, sarfati2025post}.} In the context of robust Bayesian analysis, the issue of dynamic (in)consistency has been studied extensively across both the statistical and axiomatic decision theory literatures \citep[e.g.,][]{vidakovic2000gamma, ES_recursive, Klibanoff_Hanany_2007_TE, Klibanoff_Hanany_2009_BEJTE, Stoye_2012_AR, fernandez2024robust, Lim_2026_PIA}. By contrast,  the dynamic inconsistency of frequentist optimality criteria has received far less attention, likely because the corresponding interim problem is not well-defined. We address this gap by formalizing the interim problem associated with such criteria, while proposing a diagnostic  for determining whether ex ante optimal decision rules provide credible prescriptions in the interim period.

More broadly, we contribute to work at the intersection of statistical and axiomatic decision theory. As advocated for in \cite{Stoye_2012_AR}, one principled approach to selecting an optimality criterion is axiomatic: the practitioner first considers the behavioral characterizations of the preferences induced by each criterion, then selects the criterion whose implications appear the most intuitive. In this regard, \cite{Hayashi_2008_JET}, \cite{Stoye_2011_JET,Stoye_2011_TD, Stoye_2012_AR}, and \cite{andrews2026misspecification} provide axiomatic foundations for various statistical decision criteria. We contribute to this literature by characterizing the behavioral foundations of our dynamically consistent optimality criteria, which correspond to the choices of a sophisticated DM who anticipates her interim preferences.

The remainder of the paper is organized as follows. Section~\ref{section_model} formally introduces the model.  Section~\ref{section_interim} presents our framework for studying the interim problem associated with frequentist minimax criteria and proposes the diagnostic procedure for assessing their interim credibility. Sections~\ref{section_treatment_choice} and \ref{section_Yata} present examples of this diagnostic exercise. Section~\ref{section_DC_criteria} defines and axiomatizes our class of dynamically consistent optimality criteria. Section~\ref{section_empirics} compares the interim prescriptions of various optimality criteria using both real and simulated data.  Section~\ref{section_conclusion} concludes.  All proofs are in the appendix, unless stated otherwise. 

\section{Model} \label{section_model}

\subsection{Preliminaries}

We first introduce the statistical decision problem in general form. Let $\Theta$ denote the state space, $ Z$ the data (sample) space, and $\mathcal A$ the action space. Unless otherwise stated, these spaces are endowed with the Borel $\sigma$-algebras generated by their respective topologies, which will be made explicit in specific applications. Let $\Delta(\cdot)$ denote the set of Borel  measures over the relevant space, and let $P_\theta\in\Delta (Z)$ be the sampling distribution under $\theta\in\Theta$, with corresponding density $p_\theta$. 

A decision rule is a measurable map $\delta: Z\to\Delta(\mathcal A)$, where $\mathcal A$ is an action space and $\Delta(\mathcal A)$ allows for randomized actions. Let $\mathcal D$ denote the set of feasible decision rules, which we assume is closed under measurable pasting.\footnote{That is, for any $\delta_1,\delta_2\in\mathcal D$ and measurable $B\subseteq Z$, the rule that agrees with $\delta_1$ on $B$ and with $\delta_2$ on $B^c$ also belongs to $\mathcal D$. This condition is satisfied, for instance, when $\mathcal{D}$ is the collection of all decision rules $\delta:  Z \to \Delta(\mathcal{A})$.} Given the payoff (i.e., negative loss)  function $u:\mathcal A\times\Theta\to\mathbb R$, we write $u(\delta( z),\theta)$ to denote the payoff averaged over the randomized action $\delta( z)$. The expected payoff (i.e., negative risk) from rule $\delta$ in state $\theta$ is $U(\delta,\theta)=\E_\theta[u(\delta( z),\theta)]$, where $\E_\theta$ denotes integration against $P_\theta$. Regret is defined relative to the optimal  feasible rule given oracle knowledge of $\theta$, such that  $R(\delta,\theta) \equiv \sup_{d\in\mathcal D}U(d,\theta)-U(\delta,\theta).$

 \subsection{Ex Ante Decisions}

In  Bayesian decision theory, the DM is assumed to have a prior  $\pi \in \Delta(\Theta)$. The corresponding ex ante expected payoff from  $\delta$ is $\E_\pi[U(\delta,\theta)]$, where $\E_\pi$ denotes integration against the prior $\pi$. The Bayes regret of $\delta$ is $\E_\pi[R(\delta,\theta)] = \E_\pi[\sup_{d\in\mathcal D} U(d,\theta)]-\E_\pi[U(\delta,\theta)]$. Since the first term does not depend on $\delta$, minimizing Bayes regret is equivalent to maximizing  expected payoff. We refer to any such rule as the Bayes rule under prior $\pi$, denoted $\delta_\pi\in\arg\max_{\delta\in\mathcal D}\E_\pi[U(\delta,\theta)]$.

The robust Bayesian generalization allows for multiple priors.\footnote{See \cite{Giacomini_Kitagawa_Read_2025} for a recent review of robust Bayesian analysis.}   Let $\Pi \subseteq\Delta(\Theta)$ denote the set of priors the DM regards as feasible. The Gamma-minimax loss (written in payoff form) and regret rules solve, respectively,  
\begin{align*}
    \delta_{\Pi\text{-MML}} & \in\arg\max_{\delta\in\mathcal D}\inf_{\pi\in\Pi}\E_\pi[U(\delta,\theta)] \\  
    \delta_{\Pi\text{-MMR}} &  \in\arg\min_{\delta\in\mathcal D}\sup_{\pi\in\Pi}\E_\pi[R(\delta,\theta)].
\end{align*}
The distinction between loss- and regret-based rules is salient with multiple priors because the oracle term $\E_\pi[\sup_{d\in\mathcal D}U(d,\theta)]$ varies with $\pi$.  With singleton  $\Pi = \{\pi\}$, both rules reduce to the Bayes rule under $\pi$.

Two canonical frequentist optimality criteria can be obtained as special cases of the Gamma-minimax criteria.  Minimax loss (i.e., maximin payoff) evaluates rules by their worst-case payoff across states and solves
\begin{align*}
    \delta_{\mathrm{MML}}\in\arg\max_{\delta\in\mathcal D}\inf_{\theta\in\Theta}U(\delta,\theta),
\end{align*}
while minimax regret does the same using regret to solve
\begin{align*}
    \delta_{\mathrm{MMR}}\in\arg\min_{\delta\in\mathcal D}\sup_{\theta\in\Theta}R(\delta,\theta). 
\end{align*}
These coincide with the Gamma-minimax criteria when $\Pi=\Delta(\Theta)$, since  worst-case expectations over $\Delta(\Theta)$ reduce to worst-case evaluations over states. In this sense, minimax loss and regret can be viewed as fully prior-free. They are also maximally conservative in the sense that  no restrictions are posed on the collection of distributions over $\Theta$,  and each rule is evaluated at the least favorable state.

\subsection{Interim Decisions}

Let $\mathcal D( z)\equiv \{\delta( z)\in\Delta(\mathcal A):\delta\in\mathcal D \}$ denote the set of feasible actions after the realization of $z \in Z$.  At each $z$, the DM's interim problem concerns the choice of an optimal action $a\in\mathcal D( z)$.  An optimality criterion is \emph{dynamically consistent} if, after observing $z$, the DM still finds it optimal to implement the action prescribed by the ex ante optimal decision rule.


\begin{definition}[Dynamic Consistency]
Fix an ex ante criterion and an associated interim criterion. Let $\mathcal D^*$ be the set of ex ante optimal rules, and let $\mathcal A^*( z)\subseteq\mathcal D( z)$ be the set of interim optimal actions after observing $ z$. The criterion is dynamically consistent if for every $\delta^* \in\mathcal D^*$, $\delta^*( z)\in\mathcal A^*( z)$ almost surely with respect to the relevant predictive distribution over $ Z$. 
\end{definition}

In the Bayesian setting, the DM's interim problem after observing the realization of $ z\in Z$ is well-defined. Given the prior $\pi$, Bayes rule induces a posterior $\pi(\cdot\mid z)$, and the DM chooses $a_{\pi(\cdot\mid z)}$ to maximize the conditional expected payoff, such that $a_{\pi(\cdot\mid z)} \in \mathcal{A}^*( z) = \argmax_{a\in  \mathcal{D}( z)} \E_{\pi(\cdot\mid z)}[u(a,\theta)]$.  As in the ex ante case, conditional expected regret yields the same optimal action.   It is well-known that the Bayes rule $\delta_\pi$ is dynamically consistent with interim problem (see, e.g., Chapter 4 of \cite{Berger_1985_SDT}). That is, we have that $\delta_\pi( z) \in \mathcal{A}^*( z)$  $P_\pi$-almost surely, where $P_\pi =\int_\Theta P_\theta \,  d\pi(\theta)$ is the prior predictive  over $ Z$.

With multiple priors, each $\pi\in \Pi$ is updated prior-by-prior, generating the posterior set $\Pi_ z \equiv \{\pi(\cdot\mid z):\pi\in\Pi \}$. The conditional Gamma-minimax loss  (written in payoff form) and regret criteria solve, respectively, 
\begin{align*}
    a_{\Pi_ z\text{-MML}} & \in\arg\max_{a\in \mathcal{D}( z) }\inf_{\pi\in\Pi_ z}\E_\pi[u(a,\theta)]   \\
    a_{\Pi_ z\text{-MMR}} & \in\arg\min_{a\in \mathcal{D}( z) }\sup_{\pi\in\Pi_ z}\E_\pi \left[\sup_{a'\in \mathcal{D}( z)}u(a',\theta)-u(a,\theta) \right].
\end{align*}
Unlike in the single-prior setting, Gamma-minimax rules are not generally dynamically consistent \citep[e.g.,][]{vidakovic2000gamma, ES_recursive}.\footnote{\cite{Lim_2026_PIA} clarifies the source of this dynamic inconsistency and proposes an alternative ex ante criterion that restores dynamic consistency. We return to this point in Section~\ref{section_DC_criteria}.}

For the frequentist criteria, we face a more fundamental difficulty: there is no canonical interim problem to compare the ex ante rule against. Minimax loss and minimax regret can be represented as $\Gamma$-minimax criteria with $\Pi=\Delta(\Theta)$, but updating the full simplex prior-by-prior does not yield a meaningful Bayesian analogue. If the realized data has positive likelihood/density at every state, then the prior-by-prior update of $\Delta(\Theta)$ remains entire simplex, i.e., no learning from the data occurs.\footnote{This would be the case, for instance, if $Z \sim N(\theta, I_m)$, such that $p_\theta(z) > 0$ for all $z \in \R^m$.} Otherwise, prior-by-prior updating is not well-defined, since there exist Dirac priors under which the realized data have probability zero. 

In the following section, we formalize the interim problem that is dynamically consistent with frequentist ex ante optimality criteria. We thereby provide a framework for  evaluating the interim credibility of decision rules selected for their ex ante frequentist minimax guarantees.


\section{Interim Problem for Frequentist Minimax} \label{section_interim}

\subsection{Preliminaries}

When solving for minimax optimal rules, a common proof strategy interprets the decision problem as a zero-sum game against Nature and invokes tools from game theory to solve for equilibrium.\footnote{See, e.g., the discussion in Section 4.1 of \cite{Stoye_2012_AR} and Chapter 5 of \cite{Berger_1985_SDT}.} In the game, Nature chooses the distribution $ \pi \in \Delta (\Theta)$ over states, while the DM  chooses $\delta\in\mathcal D$.  Under minimax loss and regret, respectively, Nature's payoff is $-U(\delta,\theta)$ and  $R(\delta, \theta) = \sup_{d\in\mathcal D}U(d,\theta)-U(\delta,\theta)$.  In either case, we then see that Nature's payoff is of the form $c(\theta)-U(\delta,\theta),$  where $c(\theta)=0$ for minimax loss and $c(\theta)=\sup_{d\in\mathcal D}U(d,\theta)$ for minimax regret.  

If the corresponding zero-sum game admits a Nash equilibrium $(\delta^*,\pi^*)$, then  $\delta^*$ is minimax optimal for the relevant criterion, and $\pi^*$ is a \emph{least favorable prior} in the sense that 
\begin{align*}
    \delta^*  & \in \arg\min_{\delta\in\mathcal D} \E_{\pi^*}  \left[ c(\theta)-U(\delta,\theta) \right]   \\
    & = \arg\max_{\delta\in\mathcal D} \E_{\pi^*} \left[  U(\delta,\theta)  \right].
\end{align*}
Note that the latter equality holds because  $c(\theta)$ does not depend on the choice of $\delta$.  In other words, we see that once a least favorable prior $\pi^*$ has been identified, both the minimax loss and minimax regret rules are Bayes against $\pi^*$. 

\subsection{Updating Least Favorable Priors}

The least favorable prior representation suggests a natural candidate for defining the interim criterion: update  $\pi^*$ and solve the corresponding conditional Bayes problem. Under the game-theoretic interpretation, this is equivalent to Nature's initial choice of $\pi^* \in \Delta(\Theta)$ being irrevocable, such that the remaining uncertainty after the data realization is reflected by its posterior update.\footnote{This idea resembles the dynamically consistent updating rules studied by \cite{Klibanoff_Hanany_2007_TE,Klibanoff_Hanany_2009_BEJTE}, which preserve the optimality of ex ante decisions after the realization of information. The distinction is that a least favorable prior is a global saddle-point object for the ex ante minimax problem, so their posterior updates need not  coincide with the full set of posteriors that support dynamic consistency.}

Let $P_{\pi^*} \equiv \int_\Theta P_\theta \, d\pi^*$ denote the prior predictive distribution over $ Z$ induced by $\pi^*$.  The following proposition shows that any interim optimality criteria dynamically consistent with frequentist minimax rules agrees with the conditional Bayes problem induced by the least favorable prior. 

\begin{proposition} \label{prop_posterior_rep}
Suppose $\delta^*\in \{\delta_{\mathrm{MML}},\delta_{\mathrm{MMR}}\}$ is a minimax optimal rule chosen from $\mathcal D$, and let $\pi^*$ be a least favorable prior for the corresponding criterion. Then,
\begin{align*}
\delta^*( z)
\in
\arg\max_{a\in\mathcal D( z)}
\E_{\pi^*(\cdot\mid z)}[u(a,\theta)]
\end{align*}
$P_{\pi^*}$-almost surely.
\end{proposition}

\begin{proof}
    See Appendix~\ref{appendix_proof_LFP}.
\end{proof}

The proof of Proposition~\ref{prop_posterior_rep}  is formally analogous to the standard argument that Bayes rules are dynamically consistent. The interpretation, however, is quite different. In the Bayesian case, the prior $\pi$ is an exogenous parameter that describes the  DM's beliefs, and dynamic consistency holds for \emph{any} decision problem. By contrast, the least favorable prior $\pi^*$ need not be unique, and it is not an exogenous object. Rather, $\pi^*$ is specific to the decision problem at hand, depending on the both the optimality criterion and the available class of decision rules  $\mathcal D$.

\subsection{Discussion}

Proposition~\ref{prop_posterior_rep} should be read as a \emph{rationalization} of interim criteria that are dynamically consistent with the minimax rule, rather than as an analogue of Bayesian dynamic consistency in the usual preference-based sense \citep[e.g.,][]{epstein1993dynamically, Ghirardato_BEU}. Given a least favorable prior $\pi^*$ that supports the ex ante minimax rule, updating $\pi^*$ produces a conditional Bayes problem for which the prescribed continuation action $\delta^*( z)$ is optimal.  However, $\pi^*(\cdot\mid z)$ need not represent the DM's actual interim beliefs or preferences.\footnote{The same observation applies to Gamma-minimax criteria: an ex ante Gamma-minimax rule is supported by a least favorable prior, whose posterior provides a Bayes rationalization of the continuation action. But in the robust Bayesian setting, the prior set $\Pi$ is specified exogenously, and the DM's interim beliefs are represented by the posterior set $\Pi_ z$, rather than the posterior induced by some least favorable prior.}

This perspective suggests a way to assess the \emph{interim credibility} of decision rules justified by ex ante minimax guarantees. Given a least favorable prior $\pi^*$ that rationalizes the ex ante minimax rule $\delta^*$, one can ask whether $\delta^*( z)$ aligns with the action a  researcher would find compelling after observing $ z$. When the two disagree,  the question becomes whether it is ``reasonable" to make decisions as if $\pi^*(\cdot\mid z)$ represented the relevant interim uncertainty.\footnote{Also  related is the ML-II approach to prior selection (e.g., Chapter 3.5.4 of \cite{Berger_1985_SDT}).}
As we  discuss in Section~\ref{section_treatment_choice}, a related exercise is considered in the conclusion of \cite{Stoye_2009_JoE}.


Our diagnostic exercise is naturally most informative when the interim choice set $\mathcal D( z)$ is sufficiently rich. Otherwise, even an implausible posterior may appear reasonable simply because the set of available alternatives leaves little scope for disagreement. The next section  shows that  the collection of least favorable posteriors can indeed be ``unreasonable" across various examples. 

Finally, we note that even though least favorable priors are central to game-theoretic proofs of minimax optimality, their implications beyond this role as a proof device have received little attention. To the best of our knowledge, this paper is the first to use least favorable priors to formalize the interim problem for frequentist decision making. 

\section{Application: Treatment Choice} \label{section_treatment_choice}

\subsection{Setup}

We now specialize the model to a problem of binary treatment choice, following  \cite{Stoye_2009_JoE}. Let $T\in \{0,1\}$ denote treatment, and suppose each individual $j \in J$ has binary potential outcomes $y^j(t)\in \{0,1\}$ for $t \in \{0,1\}$.\footnote{Everything in this section can be generalized to the setting with arbitrary bounded potential outcomes. We assume binary potential outcomes for simplicity of exposition.}  Assigning treatment $t$ to an individual induces the random variable $Y_t$. Because the payoff and the sampling distribution depend on the law of $(Y_0,Y_1)$ only through its two marginal means, we take the state to be the pair of success probabilities $\theta=(\mu_0,\mu_1)\in\Theta\equiv[0,1]^2$, with $\mu_t \equiv \Pr(Y_t=1)=\E[Y_t].$   


The action space is $\mathcal A=[0,1]$, where $a\in[0,1]$ is the probability of assigning treatment $t=1$.\footnote{Equivalently, take $\mathcal{A} = \{0,1\}$ and  consider randomized decision rules $\delta : Z \to \Delta(\mathcal{A})$.} The DM observes $ z=(t_n,y_n)_{n=1}^N
\in  Z\equiv \{0,1\}^{2N}.$   A treatment rule is  then $\delta: Z\to \mathcal{A}$,  and the expected payoff of rule $\delta$ in state $\theta=(\mu_0,\mu_1)$ is
\begin{align*}
U(\delta,\theta)
=
\mu_0\left(1-\E_\theta[\delta( z)]\right)
+
\mu_1\E_\theta[\delta( z)].
\end{align*}
If $\theta$ were known, the DM would assign treatment $1$ whenever $\mu_1>\mu_0$, treatment $0$ whenever $\mu_0>\mu_1$, and would be indifferent when $\mu_0=\mu_1$. 

We consider a random-assignment design where treatment assignments are independent fair coin flips, such that $T_j$ is independent of potential outcomes and satisfies $P(T_j=1)=P(T_j=0)= \frac{1}{2}$. Conditional on $T_j=t$, the observed outcome $y_j$ is an independent realization of $Y_t$.   

\subsection{Ex Ante Decision}

 
Let $N_t$ denote the number of sample observations assigned to treatment $t$, and let $\bar y_t$ be the corresponding sample average. The following fact restates Proposition 1(ii) of \cite{Stoye_2009_JoE} in the present notation.

\begin{fact}[\cite{Stoye_2009_JoE}] \label{fact_stoye}
Let $I_N \equiv N_1(\bar y_1- \frac{1}{2})-N_0(\bar y_0- \frac{1}{2})$.  The   optimal treatment rule according to the minimax regret criterion is
\begin{align*}
\delta^*( z)
=
\begin{cases}
0, & I_N<0 \\
1/2, & I_N=0 \\
1, & I_N>0.
\end{cases}
\end{align*}
\end{fact}

The least favorable prior used in \cites{Stoye_2009_JoE} proof of the minimax optimality of $\delta^*$ is  $\pi^* = \frac{1}{2}\delta_{(a,1-a)} + \frac{1}{2}\delta_{(1-a,a)}$ for some $a > \frac{1}{2}$.  Even though the least favorable prior is not unique, least favorability imposes a common restriction.

\begin{proposition} \label{prop_treatment_LFP}
Assume $N \geq 1$. There exists a unique constant $d_N \in (0,1)$ such that every least favorable prior $\pi^*$ places probability one on states satisfying 
\begin{align*}
    |\mu_1 - \mu_0| = d_N.
\end{align*}
\end{proposition}

\begin{proof}
    See Appendix~\ref{Appendix_prop_LFP}.
\end{proof}

\subsection{Interim Decision}

By Proposition~\ref{prop_posterior_rep}, the  posterior update of any least favorable prior rationalizes the  action chosen by any interim optimality criteria that is dynamically consistent with minimax regret. As discussed in the previous section, the diagnostic exercise regarding the credibility of $\delta^*$ centers on whether this posterior provides a ``reasonable" interim rationale. In the context of our treatment choice setup, Proposition~\ref{prop_treatment_LFP}  implies that any least favorable posterior is supported on states with $|\mu_1-\mu_0|=d_N$, regardless of the realized data. 

For concreteness, consider the example in the conclusion of \cite{Stoye_2009_JoE}. Suppose $N=1100$, with $N_0=1000$ observations assigned to  $t=0$ and $N_1=100$ observations assigned to $t=1$. Assume further that $n_0=550$ and $n_1=99$, such that $\bar y_0=0.55$ and $\bar y_1=0.99$.   The estimated treatment effect is $\bar y_1-\bar y_0=0.44,$  with standard error approximately $\sqrt{
\frac{0.99(0.01)}{100}
+
\frac{0.55(0.45)}{1000}
}
\approx 0.019.$

A conventional comparison of sample success rates therefore provides overwhelming evidence in favor of treatment $1$. However, the minimax regret rule $\delta^*$ assigns treatment $0$, since
\begin{align*}
I_N = 100 \left (0.99- \frac{1}{2} \right)-1000 \left(0.55- \frac{1}{2} \right) =  -1.
\end{align*}
We therefore see that the  action $\delta^*( z) = 0$ prescribed by the ex ante optimal rule disagrees with the action that most DMs would take after observing $ z$.

The posterior rationale for the action $\delta^*(z) = 0$ is difficult to defend. Here, direct evaluation of the expression for $d_N$ in the proof of Proposition~\ref{prop_treatment_LFP} yields  $d_N\approx 0.023$ (see Example~\ref{ex_difference_treatment} in  Appendix~\ref{appendix_proof}). Any least favorable posterior then assigns probability one to states in which  $|\mu_1 - \mu_0 | \approx 0.023$. The realized data, however, estimate a treatment effect of $0.44$ with standard error less than $0.02$. Thus, every least favorable posterior is concentrated on effect sizes that are far from the empirical estimate.\footnote{In the conclusion of \cite{Stoye_2009_JoE}, a similar discussion is presented using a specific least favorable prior. Proposition~\ref{prop_treatment_LFP} allows us to extend the argument to \emph{all} least favorable priors.}

This example suggests that the ex ante minimax regret rule $\delta^*$ is unlikely to be dynamically consistent. After observing $ z$, a decision maker who evaluates the evidence using a conventional comparison of sample success rates would regard treatment $t=1$ as strongly favored. The posterior induced by any least favorable prior rationalizes the prescription $\delta^*(z) = 0$ only by concentrating mass on treatment-effect magnitudes that the realized data make highly implausible.  Hence, the minimax rule $\delta^*$ lacks credibility as a decision rule that would be adhered to in the interim period, absent external commitment technologies.

\begin{remark}[Conservatism] \label{remark_conservatism}
    Minimax-regret optimal treatment rules are often less ``conservative'' than conventional frequentist procedures, such as hypothesis testing at the $\alpha = 0.05$ level. \cite{Manski_2004_ECTA} interprets this as illustrating the inadequacies of the latter approach to statistical decision making. Our example has the opposite character: the disagreement with conventional frequentist procedures arises from the minimax-regret optimal rule being ``too conservative."
    
    Although both are instances of dynamic inconsistency within our framework, their interpretations differ sharply.  Given the degree of conservatism already built into minimax regret as a worst-case, prior-free criterion, dynamic inconsistencies that arise in the former  setting point to a deficiency of the interim frequentist procedure. By contrast,  the concern we raise here is that minimax regret may be too conservative in certain settings, and its interim prescriptions can only be rationalized by “unreasonable” beliefs. In this case, the dynamic inconsistency  casts doubt on the credibility of ex ante procedure instead. 

    Figure~\ref{fig_treatment_disagreement} illustrates that the latter form of disagreement we discuss in this section occurs with non-negligible probability, particularly when the average potential outcomes $\mu_0,\mu_1$ are close to 0 or 1. Section~\ref{section_empirics} also provides examples of this phenomenon occurring in empirical practice.
\end{remark}

\begin{figure}[t]
\centering
\includegraphics[width=0.65\linewidth]{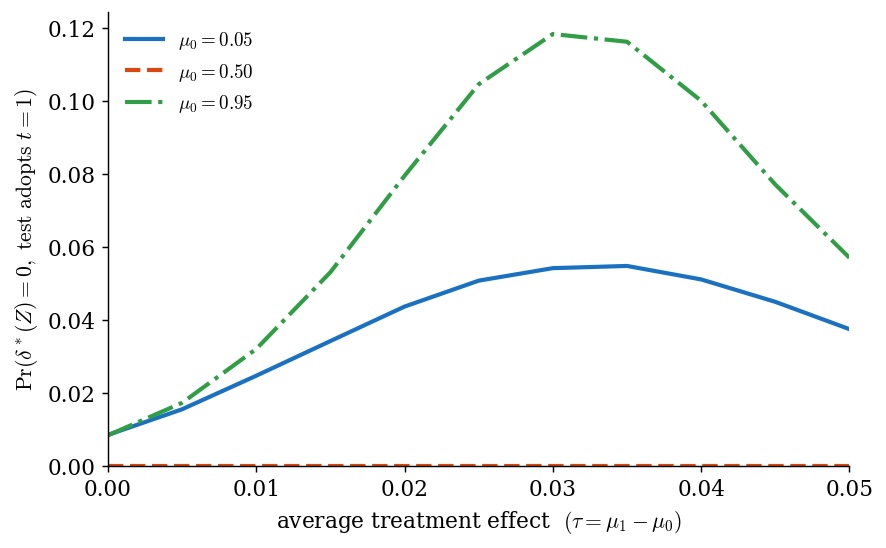}
\caption{Probability that the minimax-regret rule assigns $t = 0$, while a one-sided test  with $\alpha = 0.05$ assigns $t=1$. The $x$-axis is the average treatment effect, and each curve fixes the mean potential outcome $\mu_0$. Calculation uses $N = 1000$ observations under independent treatment assignment and $10^6$ Monte Carlo draws.}
\label{fig_treatment_disagreement}
\end{figure}

\section{Application: Evidence Aggregation} \label{section_Yata}

\subsection{Setup}

Next, we consider the evidence aggregation problem, following \cite{ishihara2021evidence} and \cite{yata2021optimal}.  A DM must decide whether to introduce a new policy in a target population. There are two study populations, indexed by $i \in \{1,2\}$, and the DM observes $ z=( z_1, z_2)\in Z\equiv\mathbb R^2$. The two studies are drawn independently, so under state $\theta=(\theta_1,\theta_2,\theta_3)$, we have $ z \sim N \left((\theta_1,\theta_2), \sigma^2 I_2\right).$  

Here, $\theta_i$ is the welfare effect of the policy in study population $i \in \{1,2\}$, while $\theta_3$ is the welfare effect in the target population. Consider the state space $\Theta = \{ \theta\in\mathbb R^3:
|\theta_1-\theta_3|\leq C_1,\
|\theta_2-\theta_3|\leq C_2 \}$,  where  $C_i$ are known constants that measure how different study population $i$ may be from the target population.  The action space is $\mathcal A=[0,1]$, where $a\in[0,1]$ is the probability of introducing the policy, and a decision rule is the mapping $\delta: Z\to\mathcal A$. Normalizing the payoff of  not introducing the policy to zero, the expected payoff of $\delta$ at $\theta$ is
\begin{align*}
U(\delta,\theta)
=
\theta_3\E_\theta[\delta( z)].
\end{align*}
If $\theta$ were known, the DM would introduce the policy whenever $\theta_3>0$, reject it whenever $\theta_3<0$, and be indifferent when $\theta_3=0$.

\subsection{Ex Ante Decision}

Suppose $C_1 > C_2$, such that study population $2$ is more externally valid than study population $1$, and define $\epsilon^* = \arg\max_{\epsilon\geq 0} (\epsilon+C_2)\Phi\left(-\frac{\epsilon}{\sigma}\right) $. The following fact restates the application of Theorem 1 to Example 1  in \cite{yata2021optimal}.

\begin{fact}[\cite{yata2021optimal}] \label{fact_yata_evidence}
If $C_1>C_2$, $2\phi(0)C_2<\sigma$, and $\epsilon^*<C_1-C_2$, then the optimal policy rule according to the minimax regret criterion is
\begin{align*}
\delta^*( z)
=
\begin{cases}
0, &  z_2<0,\\
1, &  z_2\geq 0.
\end{cases}
\end{align*}
\end{fact}

That is, the minimax regret rule ignores the first study and bases the policy decision entirely on the more externally valid study. The following result characterizes the common support restriction that all least favorable priors must satisfy. 

\begin{proposition} \label{prop_Yata_LFP}
Define   $\Theta^+ \equiv \left\{
\theta\in\Theta:
\theta_2=\epsilon^* , \,\, 
\theta_3=\epsilon^* +C_2
\right\}$ and $ \Theta^- \equiv -\Theta^+.$ Under the conditions of Fact~\ref{fact_yata_evidence}, every least favorable prior $\pi^*$ places probability one on states in $\Theta^+\cup\Theta^-$. 
\end{proposition}

\begin{proof}
    See Appendix~\ref{Appendix_prop_Yata}.
\end{proof}

As in the previous section, we see that  the class of least favorable priors share a support restriction. In particular, all least favorable priors are  supported on states in which the more externally valid study is separated from the target population by exactly $C_2$, and the target welfare effect has magnitude $|\theta_3|=\epsilon^*+C_2$. 

\subsection{Interim Decision}

By Proposition~\ref{prop_posterior_rep}, the posterior update of any least favorable prior rationalizes the action chosen by any interim optimality criterion that is dynamically consistent with minimax regret. As before, the diagnostic question is whether the class of such posteriors provides a reasonable interim rationale. In the context of our evidence aggregation problem, Proposition~\ref{prop_Yata_LFP} implies that any least favorable posterior remains supported on $\Theta^+\cup\Theta^-$, regardless of the realized data.

For concreteness, suppose that $\sigma=1$, $C_1 = 2$, and $C_2 = 1$, in which case $\epsilon^* \approx 0.132$.  Now consider the data realization  $ z_1=5$, $ z_2=-0.05$, such that the minimax rule assigns  $ \delta^*( z) = \1\{ z_2\geq 0 \} = 0$, i.e., the policy is not implemented.  

On the other hand,  most conventional analyses of the data would regard the first study as strong evidence in favor of implementing the policy. Even allowing for the external validity bound $C_1=2$, the realization $z_1=5$ suggests that the  effect  in the target population is positive by a wide margin. For example, a conservative lower confidence calculation yields
\begin{align*}
    z_1-1.96 \sigma - C_1 = 1.04 > 0,
\end{align*}
while the second study is only slightly negative and statistically insignificant. In other words, the action $\delta^*(z)$ prescribed by the ex ante optimal rule disagrees with the action that most practitioners would take after observing $z$.

As in the previous section, we address this disagreement using the diagnostic exercise proposed in Section~\ref{section_interim}. First, observe that every least favorable posterior is supported on states that satisfy  $|\theta_3| = \epsilon^*+C_2
\approx 1.132$ and
\begin{align*}
    \theta_1 & \leq \theta_3 + C_1 \approx \begin{cases}
        3.132 & \text{if } \theta_3 \geq 0 \\
        0.868  & \text{if } \theta_3 < 0.
    \end{cases}
\end{align*}
However, the realized data from the first study (a random sample from $N(\theta_1, 1)$)  is $ z_1 = 5$. The collection of least favorable posteriors then rationalizes rejecting the policy by treating the first study as a tail event nearly two standard deviations above the largest admissible value of $\theta_1$, making the posterior rationale for the minimax action  difficult to defend. 

This example  illustrates that the minimax regret rule  $\delta^*$ is unlikely to be dynamically consistent. A DM who regards $ z_1=5$ as compelling evidence about the target population would have strong reason to deviate from the ex ante minimax prescription and choose $a=1$ instead. As in the treatment choice example, we see that minimax rule is ``too conservative" and lacks credibility as a decision rule that would be adhered to in the interim period.

\section{Dynamically Consistent Optimality Criteria} \label{section_DC_criteria}

\subsection{Preliminaries}


The previous sections illustrated that decision rules justified by ex ante minimax guarantees may fail to provide credible interim prescriptions.  In this section, we address this issue by proposing a class of ex ante optimality criteria that are dynamically consistent with their corresponding interim problems. We also axiomatize these criteria, thereby providing a behavioral characterization of the preferences and choice correspondences they induce.
 


 Let $\mu\in\Delta( Z)$ denote the marginal distribution over the data space. If the DM has a prior $\pi\in\Delta(\Theta)$, then $\mu$  corresponds to the prior predictive $P_\pi\equiv\int_\Theta P_\theta d\pi(\theta)$ over $ Z$.  For each realization of $ z\in Z$, let $\mathcal{Q}_z \subseteq\Delta(\Theta)$ denote the collection of ``beliefs" the DM anticipates using at the interim stage, assuming without loss that each $\mathcal{Q}_z$ are closed and convex. We emphasize that  the set $\mathcal{Q}_z$ need not have a Bayesian interpretation. It may, for instance, correspond to the confidence region or identified set  for the parameter of interest.

We now have sufficient notation to formally introduce our two dynamically consistent optimality criteria. Their  game-theoretic interpretation is that after the DM chooses $\delta$ and the data $z$ realizes, a malevolent Nature chooses the worst-case distribution over $\Theta$ from $\mathcal{Q}_z$. 

\begin{definition}[Dynamically Consistent Minimax Loss]
A rule $\delta_{\mathrm{DC\text{-}MML}}$ is a dynamically consistent minimax loss (maximin payoff) rule if
\begin{align*}
    \delta_{\mathrm{DC\text{-}MML}}  \in  \argmax_{\delta\in\mathcal D}   \E_{\mu}
    \left[
    \inf_{\pi\in\mathcal Q_ z}
    \E_{\pi}
    \left[
    u(\delta( z),\theta)
    \right]
    \right]
\end{align*}
for some $\mu\in\Delta( Z)$ and correspondence $ z\mapsto\mathcal Q_ z\subseteq\Delta(\Theta)$.
\end{definition}

\begin{definition}[Dynamically Consistent Minimax Regret]
A rule $\delta_{\mathrm{DC\text{-}MMR}}$ is a dynamically consistent minimax regret rule if
\begin{align*}
\delta_{\mathrm{DC\text{-}MMR}}
\in
\argmin_{\delta\in\mathcal D}
\E_{\mu}
\left[
\sup_{\pi\in\mathcal Q_ z}
\E_{\pi}
\left[
\sup_{a'\in\mathcal D( z)}
u(a',\theta)
-
u(\delta( z),\theta)
\right]
\right]
\end{align*}
for some $\mu\in\Delta( Z)$ and correspondence $ z\mapsto\mathcal Q_ z\subseteq\Delta(\Theta)$.
\end{definition}


Both criteria are dynamically consistent by construction. Conditional on observing $ z$, the corresponding interim minimax loss and regret  problem is to solve
\begin{align*}
a_{\mathcal{Q}_ z\text{-}MML} & \in  \argmax_{a\in\mathcal D( z)}
\inf_{\pi\in\mathcal Q_ z}
\E_{\pi}
\left[
u(a,\theta)
\right] \\ 
a_{\mathcal{Q}_ z\text{-}MMR} & \in  \argmin_{a\in\mathcal D( z)}
\sup_{\pi\in\mathcal Q_ z}
\E_{\pi}
\left[
\sup_{a'\in\mathcal D( z)}
u(a',\theta)
-
u(a,\theta)
\right],
\end{align*}
and the ex ante function aggregates these interim objectives using the marginal $\mu$. 

\subsection{Discussion}

A DM making decisions according to either criteria is \emph{sophisticated}, in the sense of \cite{Strotz_ReStud_1955}. She anticipates, at the ex ante stage, the procedure she will use in the interim stage after the data realization, and she chooses a decision rule whose prescriptions are optimal according to that anticipated interim procedure.

The marginal distribution $\mu$ can be interpreted in two ways. In the classical Bayesian setting, the DM has a prior $\pi\in\Delta(\Theta)$, and $\mu = P_\pi$ is the prior predictive  over $Z$. On the other hand, the DM may have subjective beliefs directly over $Z$ instead, without having a prior over $\Theta$. As discussed in the introduction, the latter scenario is more realistic when  $\Theta$ is high-dimensional or lacks a concrete interpretation, while the empirical phenomenon itself is familiar enough to the DM to support judgments about the marginal distribution of sample outcomes.\footnote{See Chapter 3.5.2 of \cite{Berger_1985_SDT} for a more detailed discussion on the various sources of information about the marginal distribution $\mu$.} Our optimality criteria only require the existence of a subjective marginal distribution $\mu$, together with an anticipated interim procedure induced by the mapping $z \mapsto\mathcal \mathcal{Q}_z$.  Bayesian decision making is then a  special case with $\mu=P_\pi$ and $\mathcal Q_ z= \pi(\cdot\mid z)$.

In the following two examples, we show that the as-if optimization of \cite{Manski_2021_ECTA} and Gamma${}^*$-minimax criteria of \cite{Lim_2026_PIA} can be nested as special cases of our optimality criteria. For illustrative purposes, we use minimax regret in the former and  minimax loss in the latter; analogous arguments also hold in reverse.

\begin{example}[As-if Optimization]  In what he calls  ``as-if decisions  with a set estimate," \cite{Manski_2021_ECTA} formalizes a common form of plug-in decision making.\footnote{Recent work that study and operationalize this framework include \cite{Andrews_Chen_2025}, \cite{Chernozhukov_et_al_2025}, and \cite{ben2025safe}.} After observing $ z$, the DM constructs a data-dependent set of states $\hat\Theta( z)\subseteq\Theta$, and chooses an action \emph{as if} the true state lies in $\hat\Theta(z)$. When decisions are only based on point estimates, $\hat\Theta( z)$ is a singleton.   In payoff notation, the as-if minimax action is defined pointwise by
\begin{align*}
\delta_{a\text{-}MMR}( z)
\in
\argmin_{a\in\mathcal D( z)} \sup_{\theta \in \hat \Theta(z)}
\left\{ \sup_{a' \in D(z)}
u(a',\theta) - u(a,\theta) \right\}. 
\end{align*}
For instance, if $\hat\Theta( z)$ is a confidence region satisfying $\inf_{\theta\in\Theta}
P_\theta \{\theta\in\hat\Theta( Z)\}
\geq
1-\alpha,$  then $\delta_{a\text{-}MMR}$ minimizes the worst-case regret over states not rejected at level $\alpha$. 

Observe that the interim problem given by the as-if criterion can be nested by our dynamically consistent minimax regret criterion by taking $\mathcal Q_z=\Delta(\hat\Theta( z)),$  where $\Delta(\hat\Theta( z))$ is the  simplex supported on $\hat\Theta( z)$. For any $a\in\mathcal D( z)$, we have $\sup_{\pi\in\mathcal Q_ z}
\E_\pi\left[ \sup_{a' \in \mathcal D( z)} u(a',\theta) - u(a,\theta) \right]
=
\sup_{\theta \in \hat \Theta( z)} \left\{ \sup_{a' \in \mathcal D( z)} u(a',\theta) - u(a,\theta) \right\},$  and the ex ante criterion becomes
\begin{align*}
\delta_{\mathrm{DC\text{-}MMR}}
\in
\argmin_{\delta \in \mathcal{D}} \E_\mu \left[ \sup_{\theta \in \hat \Theta( z)}   
\left\{ \sup_{a' \in \mathcal D( z)}
u(a',\theta) - u(\delta( z),\theta) \right\} 
\right], 
\end{align*}
such that $\delta_{\mathrm{DC\text{-}MMR}}$ selects an as-if minimax action, $\mu$-almost surely. 
\end{example}

\begin{example}[Gamma${}^*$-Minimax]

In the setting of robust Bayesian analysis,  \cite{Lim_2026_PIA} proposes the Gamma${}^*$-minimax criterion to address the dynamic inconsistency of Gamma-minimax rules.  Let $\Pi \subseteq\Delta(\Theta)$ be an ex ante set of priors. Under the consistency condition in \cite{Lim_2026_PIA}, every prior $\pi\in  \Pi$ induces the same prior predictive $P_\pi = \int_\Theta P_\theta \, d\pi(\theta)$ over $ Z$, in which case we can set $\mu \equiv P_\pi$ without loss.

After observing $z$, the DM updates prior-by-prior to obtain the posterior set $\Pi_z$. The corresponding interim problem is the conditional Gamma-minimax problem,
\begin{align*}
    \delta_{\Gamma^*\text{-MML}} ( z) \in \argmax_{a \in \mathcal{D}( z)} \inf_{\pi \in \Pi_ z} \E_\pi [ u(a,\theta)].
\end{align*}
The Gamma${}^*$-minimax criterion first defines the set of joint distributions over $\Theta \times Z$ given by  $\Gamma^* =
\left\{
\gamma\in\Delta(\Theta\times Z): \gamma(d\theta, d z) = \mu(d z) \pi_ z ( d\theta) \right\}$, where $z \mapsto \pi_z$ is a measurable selection satisfying $\pi_z \in \Pi_z$,  $\mu$-almost surely. Then, it solves 
\begin{align*}
\delta_{\Gamma^*\text{-}MML}
& \in
\argmax_{\delta\in\mathcal D}
\inf_{\gamma\in\Gamma^*}
\E_\gamma
\left[
u(\delta( z),\theta)
\right] \\ 
& = \argmax_{\delta \in \mathcal{D}} \E_\mu \left[
\inf_{\pi\in\Pi_ z}
\E_\pi \left[
u(\delta( z),\theta)
\right] \right],
\end{align*}
which coincides with $\delta_{\text{DC-MML}}$ with $\mathcal{Q}_z = \Pi_z$.\footnote{Formally, the identity holds because  $\Gamma^*$ is rectangular, in the sense of \cite{ES_recursive}, with respect to the partition $\{ \Theta \times \{z\} \}_{z \in Z}$.} Here, the interim belief set $\mathcal{Q}_z$  carries an explicit  Bayesian interpretation as the prior-by-prior update of $\Pi$. 
\end{example}

\subsection{Axiomatics}

We now turn to providing a behavioral characterization of our  two optimality criteria. Fixing some utility function $u: \mathcal{A} \times \Theta \to \R$, we abuse notation by associating each decision rule $\delta : Z \to \mathcal{A}$ with the \emph{utility act} $\delta \in \R^{Z \times \Theta}$ defined by $\delta(z,\theta) = u(\delta(z), \theta)$.\footnote{The utility function can be recovered from preferences if we instead associate decision rules with Anscombe-Aumann acts that map from $Z \times \Theta$ to some mixture space.} If $\delta( \cdot, \theta) = \delta(\cdot, \theta')$ for all $\theta, \theta' \in \Theta$, we say that the utility act $\delta$  is \emph{Z-measurable}. Likewise, if $\delta(z, \cdot) = \delta(z', \cdot)$ for all $z,z' \in Z$, then $\delta$  is \emph{$\Theta$-measurable}.  Throughout the remainder of this section, we assume for simplicity that both $Z$ and $\Theta$ are finite.

Most axiomatizations of statistical optimality criteria proceed by taking preferences or choice correspondences over \emph{risk acts} as the primitive.\footnote{See, e.g., \cite{Stoye_2012_AR} for a review. A recent exception is the axiomatic work of \cite{andrews2026misspecification}, who consider the full space of utility (i.e., loss) acts for a different reason.}  That is, they fix a likelihood function $\ell: \Theta \to \Delta(Z)$ and associate each decision rule with the (negative) risk act $f_\delta \in \R^\Theta$ defined by $f_\delta(\theta) = \E_{\ell(\cdot|\theta)} [u(\delta(Z), \theta)]$.   This construction reduces decision rules to  $\Theta$-measurable  acts. As observed by \cite{Lim_2026_PIA}, this reduction is \emph{with} loss when considering dynamically consistent optimality criteria.

 \subsubsection{Dynamically Consistent Minimax Loss}

To axiomatize the dynamically consistent minimax loss criterion, we take the preference $\succsim$ over utility acts in $\mathcal{D} \cong \R^{Z \times \Theta}$ as the primitive relation.  

Our first axiom requires that $\succsim$ satisfies the subjective expected utility (SEU) axioms over the collection of $Z$-measurable acts.  Intuitively, this will allow us to recover a unique marginal $\mu \in \Delta(Z)$ that captures the DM's subjective beliefs about the marginal distribution over the data space. We also assume for simplicity that the recovered SEU measure $\mu$ has full support.

\begin{axiom}[$Z$-SEU] \label{axiom_seu}
    $\succsim$ is full-support SEU over $Z$-measurable acts. 
\end{axiom}

For $\delta, \beta \in \mathcal{D}$ and some $z' \in Z$, define $\delta_{ \{z' \} } \beta$ as the spliced act that satisfies $ (\delta_{ \{z' \} } \beta ) (z, \theta) =  \delta(z,\theta)$ when $z = z'$ and $ (\delta_{ \{z' \} } \beta ) (z, \theta) = \beta(z, \theta)$  when $z \neq z'$.  The next axiom imposes singleton separability  with respect to this splicing operation over the data space $Z$.

\begin{axiom}[$Z$-Separability] \label{axiom_z_sep_loss}
            $\succsim$ satisfies $Z$-Separability if for every $\delta, \beta, \psi, \phi \in \mathcal{D}$,
        \begin{align*}
            \delta_{ \{z\} } \psi \succsim  \beta_{ \{z\} } \psi \hspace{10pt} \iff \hspace{10pt} \delta_{ \{z\} } \phi \succsim \beta_{ \{z\} } \phi 
        \end{align*}
\end{axiom}

When comparing two acts that map to different actions only when the data realization $\{Z=z\}$ occurs, the $Z$-separability axiom requires that the ranking must not depend on what happens outside of $\{Z=z\}$. It can thus be interpreted as a weakening of the Sure Thing Principle, imposed only over data realizations. Equivalently, if the DM anticipates evaluating actions after observing $z$, then prescriptions at unrealized data values are counterfactual and should not affect her comparison of the actions prescribed at $z$.

Let $\succsim_z$ be the induced preference over $\Theta$-measurable acts defined by $\delta \succsim_z \beta \Leftrightarrow \delta_{ \{z\} } \phi \succsim \beta_{ \{z\} } \phi$, where $\phi$ is any utility act. Note that $\succsim_z$ is well-defined under $Z$-separability, and it can naturally be interpreted as the  conditional preference relation given the realization of $z$. The final axiom requires that every $\succsim_z$ satisfies the minimax expected utility (MEU) axioms of \cite{GS_MEU}. 

\begin{axiom}[$\Theta$-MEU] \label{axiom_meu}
   $\succsim_z$ is MEU for every $z \in Z$.
\end{axiom}

The $\Theta$-MEU axiom provides a behavioral interpretation of the interim belief correspondence $z \mapsto \mathcal \mathcal{Q}_z$. At each realized $z$, the induced preference $\succsim_z$ behaves as if the DM evaluates $\Theta$-measurable acts by their worst-case expected utility according to some set of beliefs over $\Theta$. The axiom  permits ambiguity at the interim stage, while requiring that it only enters conditional on each data realization.

As shown by the following theorem, the above three axioms are equivalent to the DM having a  dynamically consistent minimax loss representation.

 \begin{theorem} \label{thm_rep_loss}
    $\succsim$ satisfies Axioms~\ref{axiom_seu}-\ref{axiom_meu} if and only if it admits a  dynamically consistent minimax loss representation with unique $(\mu, \{\mathcal{Q}_z\}_{z \in Z})$. 
\end{theorem}
\begin{proof}
    See Appendix \ref{appendix_proof_loss}.
\end{proof}

\subsubsection{Dynamically Consistent Minimax Regret}

It is well-known that optimality criteria based on minimax regret do not satisfy the  Independence of Irrelevant Alternatives (IIA) axiom \citep[e.g.,][]{Hayashi_2008_JET, Stoye_2011_JET}. For this reason, we must take choice correspondences $C(\cdot)$ that map finite menus $M \subsetneq \mathcal{D}$ to $C(M) \subseteq M$ as the primitive relation.  We write $C_{(z,\theta)}(M)$ to denote the choice made after learning $(z,\theta)$ has realized.\footnote{Equivalently, $C_{(z,\theta)}(M) = C(M_{z,\theta})$, where $M_{z,\theta}$ contains constant acts $\delta(z,\theta)$ with $\delta \in M$.} Mixtures between menus are defined with respect to Minkowski addition, such that $\lambda M + (1-\lambda) N = \{ \lambda m + (1-\lambda) n : m \in M, \, n \in N\}$.  We say that a menu has \emph{state-independent outcome distributions} if $\{ p \in \R : \delta(z, \theta) = p \text{ for some } \delta \in M\}$ does not vary with $(z, \theta)$. That is, the collection of feasible utilities is constant across states.

Our axiomatization of the dynamically consistent minimax regret criterion proceeds by introducing additional structure on the endogenous prior minimax regret representation of \cite{Stoye_2011_JET}. The axioms that we adopt without modification are as follows. We refer the reader to Section 2 of \cite{Stoye_2011_JET} for their interpretation.

\begin{axiom}[\cite{Stoye_2011_JET}] \label{axiom_stoye}
    $C(\cdot)$ satisfies the following conditions.
    \begin{enumerate} 
        \item \emph{(Non-triviality)}  There exists some $M \subseteq \mathcal{D}$ such that $C(M) \subsetneq M$.
        \item \emph{(Monotonicity)}  If $\delta \in M$,  $\beta \in C(M)$, and $\delta \in C_{(z,\theta)}(\{\delta,\beta\})$ for all $(z,\theta) \in Z \times \Theta$, then $\delta \in C(M)$.
        \item \emph{(Independence)} $C(\lambda M + (1-\lambda) \delta)  = \lambda C(M) + (1-\lambda) \delta$ for any $\lambda \in (0,1)$.  
    \item \emph{(Independence of Never Strictly Optimal Alternatives (INA))}      If $C_{(z,\theta)}(M\cup N) \cap M \neq \emptyset$ for all $(z,\theta) \in Z \times \Theta$, then $C(M\cup N) \cap M \in \{C(M), \emptyset\}.$ 
    \item \emph{(Mixture Continuity)}  Suppose $C(M \cup \{\delta\}) = \{\delta\}$ and $\delta \not \in M$. For any $\beta \in M$ and $\beta' \in \mathcal{D}$, there exists $\lambda \in (0,1)$ such that $  C(M \cup \{\lambda \delta + (1-\lambda) \beta'\} )  = \{\lambda \delta + (1-\lambda) \beta'\}$ and  $ \lambda \beta + (1-\lambda) \beta'  \not \in C(M \cup \{\lambda \beta + (1-\lambda) \beta' \} ) $.

    \item \emph{(Ambiguity Aversion)} For any $\delta, \beta \in \mathcal{D}$,  $\lambda \in [0,1]$, and $M \supseteq \{\delta, \beta, \lambda \delta + (1-\lambda) \beta\}$,   $ \{\delta, \beta \} \subseteq C(M)$ implies  $\lambda \delta + (1-\lambda) \beta \in C(M)$.  

        \item \emph{(C-betweenness)} For any $\delta \in \mathcal{D}$, $p \in \R$, $\lambda \in (0,1)$ and $M \supseteq \{p, \delta, \lambda \delta +(1-\lambda) p \}$ with state-independent outcome distributions, $  p,\delta \not \in C(M)$ implies $\lambda \delta + (1-\lambda) p \not \in C(M)$.   
    \end{enumerate}
\end{axiom}

We need two additional axioms to characterize the dynamically consistent minimax regret criterion. The first axiom imposes IIA on $Z$-measurable acts, strengthening \cites{Stoye_2011_JET} IIA for constant acts.  Intuitively, the axiom will require that $C(\cdot)$ induces a binary relation over the collection of $Z$-measurable acts, in which case monotonicity, independence, and mixture continuity axioms will allow us to recover the unique SEU measure $\mu \in \Delta(Z)$ as in the proof of Theorem \ref{thm_rep_loss}. Again, we assume for  simplicity that the recovered SEU measure $\mu$ has full support.

\begin{axiom}[IIA on $Z$-measurable Acts]
    For menus of $Z$-measurable acts $M,N$, 
    \begin{align*}
        C (M \cup N) \cap M \in \{ C(M), \emptyset\}.
    \end{align*}
\end{axiom}

Finally, we need to adapt the $Z$-Separability axiom from the case of minimax loss to the setting of choice correspondences. The following axiom does exactly this, requiring that the choice correspondence $C(\cdot)$ respect the Sure Thing Principle when splicing acts over data realizations. Its interpretation is analogous to that of $Z$-Separability in the previous section.

\begin{axiom}[$Z$-Separable Choice] \label{axiom_sep_choice}
    For every $M$ with $\delta_{ \{z\}} \phi, \delta_{ \{z\}} \phi' \in C(M)$, we have
    \begin{align*}
        \delta_{ \{z\} } \phi \in C \left( M \cup \left\{ \beta_{ \{z\} } \phi \right\} \right)   \hspace{10pt}\Leftrightarrow  \hspace{10pt}  \delta_{ \{z\} } \phi' \in C \left( M \cup \left\{ \beta_{ \{z\} } \phi' \right\} \right).
    \end{align*}
\end{axiom}

The following theorem illustrates that the above three axioms behaviorally characterize the dynamically consistent minimax regret representation.

\begin{theorem} \label{thm_rep_regret}
     $C(\cdot)$ satisfies Axioms \ref{axiom_stoye}-\ref{axiom_sep_choice} if and only if it admits a dynamically consistent minimax regret representation with unique $(\mu, \{\mathcal{Q}_z\}_{z \in Z})$.
\end{theorem}
\begin{proof}
    See Appendix \ref{appendix_proof_regret}. 
\end{proof}

\section{Comparison of Optimality Criteria} \label{section_empirics}

\subsection{Preliminaries}

 
In this section, we turn to evaluating the optimality criteria considered  thus far. Our first application uses \cites{bursztyn2020misperceived} study of information treatments on social norms to show that ex ante minimax regret can generate irrational interim prescriptions in treatment choice, while  dynamically consistent minimax regret avoids these issues. The second application does the same using simulated data in the context of an evidence aggregation problem with limited external validity. 

To evaluate the interim credibility of decision rules, we study the interim prescriptions generated by each rule. For the dynamically consistent minimax regret criterion, this means that only the correspondence $z \mapsto \mathcal{Q}_z$ is relevant. The DM's prior beliefs about the marginal distribution of data do not affect the interim action, provided that the ex ante  choice domain is the set of all decision rules. 

Specifically, we use the as-if minimax regret (as-if MMR) criterion of \cite{Manski_2021_ECTA} as the interim objective, taking  confidence sets $\hat \Theta(z)$ as set estimates of the parameter $\theta$ and setting $\mathcal{Q}_z = \Delta(\hat \Theta(z))$. This approach follows recent work that applies  \cites{Manski_2021_ECTA} as-if  framework to decision problems with confidence sets \citep[e.g.,][]{Andrews_Chen_2025, Chernozhukov_et_al_2025, ben2025safe}.  Since dynamically consistent minimax regret provides the ex ante foundations of such as-if criteria, the results in this section can also be interpreted as illustrating the appeal of the as-if framework.

\subsection{Application: Treatment Choice} \label{subsection_empirics_treatment} 

\subsubsection{Data}

\citet{bursztyn2020misperceived} study whether  Saudi Arabian men's willingness to help their wives search for jobs outside the home is constrained by misperceived social norms. The authors find that most men privately support women working outside the home yet underestimate how many other men do. In their main experiment, the authors correct this misperception for a randomly selected half of the 500 participants.\footnote{The authors electronically randomized treatment values across all one thousand three-digit combinations before the experiment and assigned each man the value corresponding to the last three digits of his phone number (\citealp{bursztyn2020misperceived}, footnote 17).} We refer to this belief correction as the \emph{information treatment} and to the no-information treatment as the \emph{status quo}.

The outcome variable is an incentivized choice at the end of the session: whether to sign one's wife up for a job-matching service instead of receiving a cash bonus. At the aggregate level, the information treatment increases the sign-up rate from $0.235$ to $0.320$, a difference of $0.085$ ($z=2.1$). To study the minimax regret treatment choice problem with covariates, we form covariate subgroups from the twelve pre-treatment characteristics in the experiment.\footnote{These include the men's elicited second-order beliefs about how many other men support women working (in general, in semi-segregated environments, and on a minimum wage), the wedge between each belief and the truth, the confidence attached to it, and education, employment, and the wife's employment.} We form cells from covariates and their two- and three-way interactions, retaining only cells with at least five participants in each treatment arm (see  Appendix~\ref{app:protocol} for details).


\subsubsection{Decision Rules}

For each covariate cell, the DM must decide whether to \emph{provide} the information treatment $(a=1)$ or \emph{withhold} it $(a=0)$ as a function of the data.\footnote{The DM may be, for example, a policymaker deciding whether to institute targeted information campaigns to correct Saudi Arabian men's misperceptions about other men's approval of women working outside the home.} Given the covariate partition $\mathcal{X}$, the DM's ex ante problem is then to choose a treatment rule $\delta: Z \to [0,1]^\mathcal{X}$. Let $z_x$ denote the data from individuals in covariate cell $x \in \mathcal{X}$, and let $\delta_x(z)$ denote the treatment probability assigned to them after the realization of $z \in Z$.  The conditional average treatment effect  $\tau_x$  denotes the effect of the information treatment on job-matching service sign-up.

Below, we consider the interim prescriptions of various decision rules after observing the realization of \cites{bursztyn2020misperceived} study.

\paragraph{Minimax Regret} The minimax regret rule from Section~\ref{section_treatment_choice} extends naturally to the setting with covariates: Proposition 3 of \cite{Stoye_2009_JoE} implies that applying the rule separately within each covariate cell is minimax-regret optimal. In other words, if $\delta$ denotes the minimax regret rule from Fact~\ref{fact_stoye}, then the rule $\delta^*: Z \to \{0,1\}^\mathcal{X}$ defined by $\delta_x^*(z) = \delta(z_x)$ for every $x \in \mathcal{X}$ is minimax regret optimal. 


\paragraph{Hypothesis Testing} A conventional one-sided size-$\alpha$ hypothesis testing rule with $H_0:\tau_x \le 0$  provides the information treatment only when there is statistically significant evidence of benefit within each covariate cell. We denote this rule by $T_\alpha$, where  $(T_\alpha)_x(z) = 1$ if the test rejects  in  cell $x$, while $(T_\alpha)_x(z) = 0$ otherwise.

\paragraph{As-If MMR} The as-if MMR rule  minimizes regret over a confidence set for the cell-specific treatment effect. Let this confidence set be $\hat \Theta_x(z)=[\hat\tau_x -c \cdot \widehat{\mathrm{se}},\hat\tau_x +c \cdot \widehat{\mathrm{se}}]$. Within each cell, the regret of providing treatment is $\max\{0,-\tau_x\}$, while the regret of withholding treatment is $\max\{0,\tau_x\}$. Over the confidence set $\hat \Theta_x(z)$, the worst-case  regrets are  $ \max\{0, -\hat\tau_x + c \cdot \widehat{\mathrm{se}}\}$ and $\max\{0,\hat\tau_x +c \cdot \widehat{\mathrm{se}}\}$, respectively. The former is strictly smaller if and only if $\hat\tau_x>0$, with indifference when $\hat\tau_x=0$. Thus, the as-if MMR rule coincides with the conditional empirical success rule of \cite{Manski_2004_ECTA}: it assigns treatment to individuals in cell $x$ if and only if the estimated treatment effect in that cell is positive, i.e.,  $\hat\tau_x=\bar Y_{1x}-\bar Y_{0x}>0$.\footnote{Note that the  as-if MMR rule is  distinct from the conditional empirical success rule in general. The evidence aggregation application we consider in Section \ref{subsection_ev_aggregation} is one such example.}

\paragraph{}  Since none of the above three rules pool data across covariate cells, we henceforth suppress the subscript $x$ and treat all decisions as cell-by-cell.

Observe that the minimax regret and as-if MMR rules are closely related. Writing $\Delta \equiv N_1-N_0$,  the minimax regret index $I_N$ from Fact~\ref{fact_stoye} can be written as $I_N=N_0\hat\tau+\Delta(\bar Y_1-\tfrac12)$. The two rules then  coincide whenever the treatment arms are balanced, i.e., $\Delta = 0$, and they differ only when the latter term overturns the sign of $\hat\tau$, i.e., $|\hat\tau|\lesssim |\Delta|\,|\bar Y_1-\tfrac12|/N_0$. Under fair-coin assignment, $|\Delta|$ is of order $\sqrt N$, so the two rules disagree only for estimated treatment effects in an $O_p(N^{-1/2})$ neighborhood of zero. Because this band is of the same stochastic order as the standard error itself, such disagreements can nevertheless coincide with statistically significant estimates.

Moreover, the following lemma illustrates the as-if MMR rule provides  treatment whenever the  hypothesis testing rule does. This eliminates the type of interim disagreement studied in Section~\ref{section_treatment_choice}, where minimax regret recommends withholding treatment despite statistically significant evidence in favor of treatment.

\begin{lemma}\label{lem:asif-Ta}
For any level $\alpha\in(0,\tfrac12)$, if $T_\alpha$ provides the information treatment, then the as-if MMR rule also provides.
\end{lemma}
\begin{proof}
See Appendix~\ref{appendix_lem_asif_Ta}.
\end{proof}

\subsubsection{Results}

We begin by focusing on disagreements between the minimax regret and hypothesis testing rules of the kind discussed in Section~\ref{section_treatment_choice}, where the minimax regret rule $\delta^*$  is ``too conservative'' and withholds treatment. As shown in Table~\ref{tab:bursztyn-ta}, across 18 distinct covariate subgroups,  $\delta^\ast$ withholds treatment even though the hypothesis testing rule $T_\alpha$ provides; 2 of these disagreements occur at the $\alpha=0.05$ level and 16 more at the $\alpha = 0.10$ level.\footnote{These subgroups are the distinct samples behind the cell counts reported in Appendix~\ref{app:disagreement}; Appendix~\ref{app:cells} explains the overlapping structure of the cells and the scope of the exercise.} Many of these disagreements arise in the presence of treatment arm imbalance (i.e., $|\Delta| \gg 0$). Moreover, as predicted by Lemma~\ref{lem:asif-Ta}, the as-if MMR rule agrees with $T_\alpha$ for every specification reported in the table.

\begin{table}[t]
\centering
\caption{Subgroups for which $\delta^\ast$ withholds while  $T_\alpha$ provides.}
\label{tab:bursztyn-ta}
\scriptsize
\setlength{\tabcolsep}{3pt}
\begin{tabular}{lcccccc}
\toprule
& & & & \multicolumn{2}{c}{$T_\alpha$} & \\
\cmidrule(lr){5-6}
Covariate cell & $N_1{:}N_0$ & $\bar Y_1/\bar Y_0$ & $\delta^\ast$ & $.05$ & $.10$ & as-if \\
\midrule
wife not emp.; few friends; low seg.\ guess & $30{:}16$ & $0.233/0.062$ & $\times$ & $\checkmark$ & $\checkmark$ & $\checkmark$ \\
employed; mid WWOH wedge; underest.\ seg. & $26{:}9$ & $0.346/0.111$ & $\times$ & $\checkmark$ & $\checkmark$ & $\checkmark$ \\
underest.\ seg.\ support; low WWOH guess; low wage guess & $23{:}12$ & $0.261/0.083$ & $\times$ & $\times$ & $\checkmark$ & $\checkmark$ \\
low educ.; low conf.; underest.\ seg.$^{\dagger}$ & $30{:}21$ & $0.167/0.048$ & $\times$ & $\times$ & $\checkmark$ & $\checkmark$ \\
no insur.\ belief; low seg.\ guess & $22{:}9$ & $0.318/0.111$ & $\times$ & $\times$ & $\checkmark$ & $\checkmark$ \\
wife not emp.; few friends; low wage guess & $29{:}20$ & $0.172/0.050$ & $\times$ & $\times$ & $\checkmark$ & $\checkmark$ \\
wife emp.; mod.\ conf.; low wage guess & $22{:}11$ & $0.273/0.091$ & $\times$ & $\times$ & $\checkmark$ & $\checkmark$ \\
low educ.; low conf.; underest.\ seg.$^{\dagger}$ & $31{:}21$ & $0.161/0.048$ & $\times$ & $\times$ & $\checkmark$ & $\checkmark$ \\
no insur.\ belief; underest.\ wage & $26{:}13$ & $0.231/0.077$ & $\times$ & $\times$ & $\checkmark$ & $\checkmark$ \\
mod.\ conf.; low wage guess & $34{:}20$ & $0.235/0.100$ & $\times$ & $\times$ & $\checkmark$ & $\checkmark$ \\
employed; high educ.; high conf. & $21{:}8$ & $0.333/0.125$ & $\times$ & $\times$ & $\checkmark$ & $\checkmark$ \\
wife not emp.; few friends; underest.\ seg. & $44{:}27$ & $0.227/0.111$ & $\times$ & $\times$ & $\checkmark$ & $\checkmark$ \\
wife not emp.; few friends; underest.\ wage & $26{:}16$ & $0.192/0.062$ & $\times$ & $\times$ & $\checkmark$ & $\checkmark$ \\
underest.\ seg.\ support; overest.\ wage; high WWOH guess & $22{:}7$ & $0.364/0.143$ & $\times$ & $\times$ & $\checkmark$ & $\checkmark$ \\
few friends; mod.\ conf.; low wage guess & $20{:}9$ & $0.300/0.111$ & $\times$ & $\times$ & $\checkmark$ & $\checkmark$ \\
no insur.\ belief; underest.\ WWOH; underest.\ wage & $22{:}10$ & $0.273/0.100$ & $\times$ & $\times$ & $\checkmark$ & $\checkmark$ \\
no insur.\ belief; underest.\ seg.\ \& wage support & $24{:}11$ & $0.250/0.091$ & $\times$ & $\times$ & $\checkmark$ & $\checkmark$ \\
underest.\ seg.\ support; mid seg.\ guess; low wage guess & $26{:}12$ & $0.231/0.083$ & $\times$ & $\times$ & $\checkmark$ & $\checkmark$ \\
\bottomrule
\end{tabular}
\par\smallskip
\footnotesize\justifying
Notes. $\bar Y_1$ and $\bar Y_0$ correspond to the cell-specific sign-up rates among those with and without the information treatment. Labels abbreviate the covariate bins (``underest.'' = underestimates others' support; ``seg.'' = semi-segregated work; ``WWOH'' = women working outside the home); full definitions are in the replication code. $\checkmark$ refers to provide, and $\times$ refers to withhold. The two rows marked $\dagger$ apply two education cutoffs to nearly the same subgroup and are therefore not independent. 
\end{table}

Consider, for instance, the first row, which corresponds to the covariate cell containing men whose wives are not employed, who know few other participants, and who place others' support for women's work outside the home in the lowest range. In this subgroup, the average sign-up rate for the job-matching service is  $0.233$  $(N_1=30)$ and $0.062$ $(N_0=16)$  for those with and without the information treatment, respectively. The corresponding one-sided hypothesis test has $p$-value $0.041$, so $T_\alpha$  at level $\alpha=0.05$ provides the treatment. By Lemma~\ref{lem:asif-Ta}, the as-if MMR rule also provides, while the minimax regret rule $\delta^*$ withholds.

As in Section~\ref{section_treatment_choice}, we interpret this disagreement as an empirically relevant instance of the interim \emph{incredibility} of $\delta^*$. After seeing evidence that is strong enough to justify information provision under both hypothesis testing and as-if MMR, a reasonable DM may wish to deviate from the ex ante minimax-regret prescription.

Figure~\ref{fig:bursztyn-cells} provides a broader view of the disagreement among the three rules. Each point is a covariate cell, plotted by cell-specific standardized arm imbalance $(N_1-N_0)/\sqrt{N}$ and Wald statistic $\hat\tau/\widehat{\mathrm{se}}$. The dashed and dotted horizontal lines mark the rejection thresholds for  $T_\alpha$ at $\alpha=0.05$ and $0.10$, respectively. 

The left panel illustrates that most disagreements between $T_\alpha$ and minimax regret arise because hypothesis testing is conservative: many blue points lie below the $T_\alpha$ threshold, so minimax regret provides treatment while $T_\alpha$ withholds. This is the phenomenon emphasized by \cite{Manski_2004_ECTA} and discussed in Remark~\ref{remark_conservatism}, highlighting the conservatism of the hypothesis testing rule. However, we also see that minimax regret can provide treatment even when the estimated treatment effect is negative. This occurs in cells with arm imbalance in favor of the status quo, such that $(N_1-N_0)/\sqrt{N} \ll 0$. Such prescriptions are difficult to justify at the interim stage: after observing a negative estimated treatment effect, the DM is nevertheless instructed to provide treatment. 

The kinds of disagreements reported in Table~\ref{tab:bursztyn-ta} occur when arm imbalance favors the information treatment, such that $(N_1-N_0)/\sqrt{N} \gg 0$. The minimax-regret  boundary where the DM is indifferent between providing and withholding treatment slopes upward with realized arm imbalance. As a result,  treatment-heavy cells in the upper-right region of the left panel can be assigned to withhold even when the Wald statistic exceeds the $T_\alpha$ threshold. 

\begin{figure}[t]
\centering
\includegraphics[width=\linewidth]{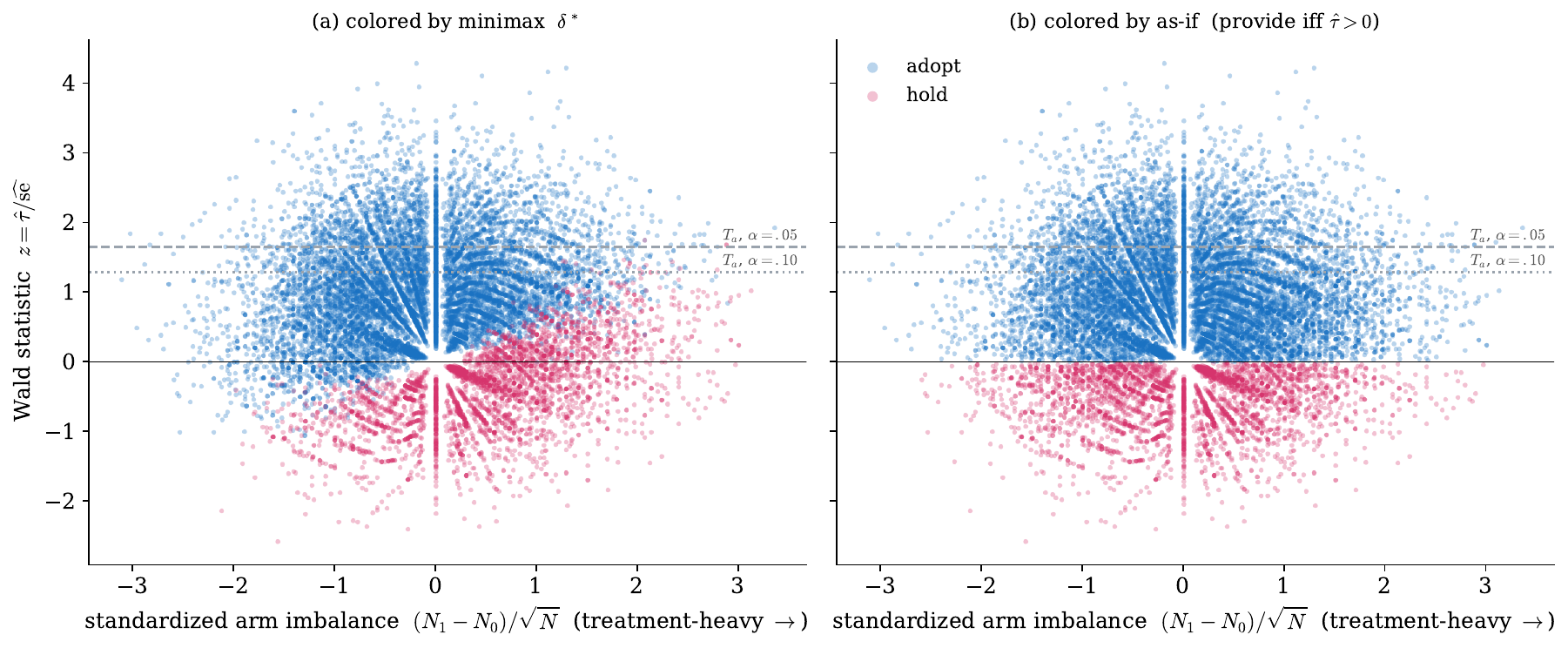}
\caption{Every covariate cell of the twelve-characteristic partition (at least five men per arm), placed by standardized arm imbalance $(N_1-N_0)/\sqrt N$ and Wald statistic $\hat\tau/\widehat{\mathrm{se}}$, colored by the minimax-regret action $\delta^\ast$ (left) and by the as-if MMR rule (right). Dashed and dotted lines denote the approval threshold for the  hypothesis testing rule $T_\alpha$ at $\alpha=0.05$ and $0.10$, respectively.} 
\label{fig:bursztyn-cells}
\end{figure}

The right panel shows that the as-if MMR rule eliminates this dependence on realized arm imbalance. As discussed previously, as-if MMR coincides with the conditional empirical success rule in this binary treatment-choice problem, and therefore provides treatment exactly when $\hat\tau>0$.  In this sense, as-if MMR avoids both the excessive aggressiveness and the excessive conservatism induced by arm imbalance under the minimax-regret rule.\footnote{Of course, the same argument holds for the conditional empirical success rule, which we know from \cite{Manski_2004_ECTA} to be asymptotically minimax-regret optimal. The evidence aggregation example we consider in the following section illustrates the value of the as-if MMR approach in settings where there is no obvious way to formalize an empirical success rule.} 



\subsection{Evidence Aggregation} \label{subsection_ev_aggregation}

\subsubsection{Decision Rules}

We now return to the evidence aggregation problem introduced in Section~\ref{section_Yata}, where the DM must decide whether to \emph{adopt} $(a=1)$ or \emph{reject} $(a=0)$ the policy after observing the realization of $z=(z_1,z_2)$. Recall that $z_i\sim N(\theta_i,\sigma^2)$, where $\theta_i$ is the welfare effect of the policy in study population $i$, and $\theta_3$ is the welfare effect in the target population. The external validity restrictions imply $|\theta_i-\theta_3|\leq C_i$ for $i\in\{1,2\}$. We will continue to maintain the assumptions stated in Fact~\ref{fact_yata_evidence}, such that study 2 is  more externally valid  than study 1.

Next, consider the interim prescriptions of the three decision rules.

\paragraph{Minimax Regret} Recall from Fact~\ref{fact_yata_evidence} that \cites{yata2021optimal}   ex ante minimax regret rule is $\delta^*(z)=\1\{z_2\geq 0\}$, i.e., only the  more externally valid study is used.

\paragraph{Hypothesis Testing} A conventional one-sided size-$\alpha$ hypothesis testing rule for $H_0: \theta_3\leq 0$ introduces the policy only when there is statistically significant evidence that the target-population welfare effect is positive. To construct the relevant lower confidence bound, first note that each study $i \in \{1,2\}$ induces the interval $[z_i-c\sigma-C_i,\, z_i+c\sigma+C_i]$, where $c$ is the appropriate critical value. Combining the two studies yields the confidence set
\begin{align*}
   [L,U] & \equiv 
    [\, z_1-c\sigma-C_1,\ z_1+c\sigma+C_1\,]
    \cap
    [\, z_2-c\sigma-C_2,\ z_2+c\sigma+C_2\,]  \\
    & = \Big[ \max\{z_1-c\sigma-C_1,\ z_2-c\sigma-C_2\}, \,\,  \min\{z_1+c\sigma+C_1,\ z_2+c\sigma+C_2\} \Big], 
\end{align*}
as long as the intersection is non-empty.\footnote{When the two study-implied intervals are disjoint, we expand the level to the smallest value of $c$ at which the intersection is nonempty and then apply the same rule to the resulting reconciled value of $\theta_3$. This keeps the rule defined for all data realizations and, in particular, keeps it using study $1$ precisely in the region where the two studies disagree and the ex ante minimax rule discards it. Such disjointness is rare: the two intervals fail to intersect only when $|z_1-z_2|>2c\sigma+C_1+C_2$. Even when the studies lie at opposite edges of their external validity bounds, this event occurs with probability $\Pr(Z>c\sqrt{2})$ ($\approx 0.3\%$ for a $95\%$ confidence set), and it is negligible otherwise.}  The hypothesis testing rule, denoted $T_\alpha$, introduces the policy if and only if $L>0$.\footnote{
Because $T_\alpha$ rejects when \emph{either} study's lower bound is positive, a single-study critical value would inflate the size of the combined test. Since the studies are independent, taking $c=\Phi^{-1}(\sqrt{1-\alpha})$ gives each study one-sided coverage $\sqrt{1-\alpha}$, so the lower bound $L=\max\{z_1-c\sigma-C_1,\,z_2-c\sigma-C_2\}$ covers $\theta_3$ with probability $1-\alpha$ and $T_\alpha$ is an exact level-$\alpha$ test of $H_0:\theta_3\le 0$. Figure~\ref{fig:yata-asif} shows the boundary $L=0$ at $\alpha=0.05$ ($c\approx1.95$) and $\alpha=0.10$ ($c\approx1.63$).} In other words, $T_\alpha$ introduces the policy only when the confidence set rules out non-positive target effects.

\paragraph{As-If MMR} The as-if MMR rule minimizes regret over the confidence set $\hat\Theta_3(z)=[L,U]$. For a given value of the target effect $\theta_3$, the regret of introducing the policy is $\max\{0,-\theta_3\}$, while the regret of rejecting the policy is $\max\{0,\theta_3\}$. Over the confidence set $[L,U]$, the worst-case regrets of adopting or rejecting the policy are then $\max\{0,-L\}$ and $\max\{0,U\}$, respectively.  The as-if MMR rule introduces the policy whenever the former is smaller than the latter. Equivalently, for a nonempty interval $[L,U]$, it introduces the policy if and only if the midpoint is positive (i.e., $(L+U)/2>0$), with indifference when $L+U=0$.

\paragraph{} Unlike the minimax regret rule $\delta^*$, the hypothesis testing and as-if MMR rules use information from both studies. The hypothesis testing rule is more conservative: it introduces the policy only when the lower endpoint $L$ is positive. By contrast, the as-if MMR rule introduces the policy whenever the confidence set is centered above zero. Hence, whenever the hypothesis testing rule adopts the policy, the as-if MMR rule does so as well, proving the following lemma.

\begin{lemma}\label{lem:yata-Ta}
For any level $\alpha \in (0,\frac{1}{2})$, if $T_\alpha$ adopts the policy, then the as-if MMR rule also adopts the policy. 
\end{lemma} 


\subsubsection{Results}

Figure~\ref{fig:yata-asif} visualizes the interim prescriptions of the three decision rules over the $(z_1,z_2)$ plane when $\sigma = 1$, $C_1 = 2$, and $C_2 = 1$. Because the minimax regret rule $\delta^*$ ignores the less externally valid study, its prescription depends only on $z_2$. It adopts the policy when $z_2 \geq 0$, i.e., in the region above the dashed black line. By contrast, the as-if MMR rule adopts the policy  when $(L+U)/2 \geq 0$, corresponding to the region above the solid black line. The hypothesis testing rule $T_\alpha$ adopts treatment when the lower confidence bound $L$ is nonnegative. In the figure, this is the upper-right region relative to the green boundaries, where the solid line corresponds to $\alpha=0.05$ and the dashed line corresponds to $\alpha=0.10$.

\begin{figure}[t]
\centering
\includegraphics[width=0.9\linewidth]{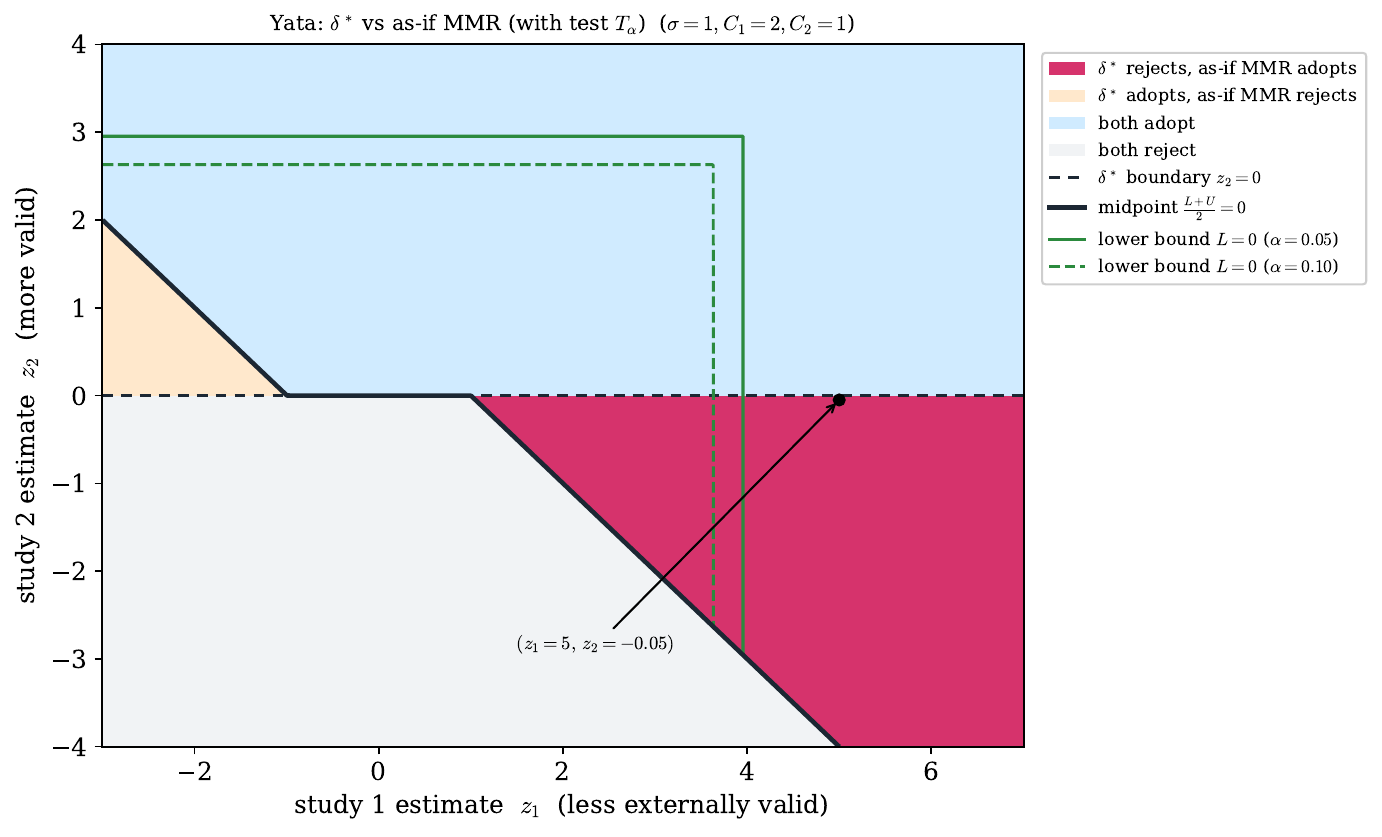}
\caption{ Color-coded regions over the $(z_1,z_2)$ plane based on the interim prescriptions of the $\delta^*$, $T_\alpha$, and as-if MMR rules  ($\sigma=1$, $C_1=2$, $C_2=1$). }
\label{fig:yata-asif}
\end{figure}

The interim \emph{incredibility} of the minimax regret rule discussed in Section~\ref{section_Yata}, where $\delta^*$ is ``too conservative'' relative to $T_\alpha$, corresponds to the region in which $T_\alpha$ adopts but $\delta^*$ rejects. In the figure, this is the pink region lying in the adoption region of $T_\alpha$ and below the dashed black line $z_2=0$.  For example, $(z_1,z_2)=(5,-0.05)$ is the data realization considered in Section~\ref{section_Yata}. Here, the minimax regret rule rejects the policy because $z_2<0$, despite significant evidence in favor of the policy from the less externally valid study.  As predicted by Lemma~\ref{lem:yata-Ta}, the as-if MMR rule adopts whenever $T_\alpha$ adopts, thus avoiding this form of excessive conservatism.

The figure also illustrates the opposite form of conservatism, in the spirit of \cites{Manski_2004_ECTA} critique that hypothesis testing can be too conservative in treatment choice. This occurs in the beige and blue regions in the interior of the green boundaries, where $T_\alpha$ rejects even though the minimax regret rule adopts. The as-if MMR rule distinguishes between the two cases. In the blue region, it adopts despite rejection by $T_\alpha$, consistent with the  prescription of $\delta^*$. In the beige region, it instead corrects the excessive aggressiveness of $\delta^*$ and rejects. Although $z_2 \geq 0$ leads $\delta^*$ to adopt, the evidence from the less externally valid study is sufficiently negative that the as-if MMR rule rejects.

Taken together, our results  illustrate that the as-if MMR rule balances the competing concerns underlying minimax regret and hypothesis testing. It avoids the excessive conservatism of $\delta^*$ when strong evidence from the less externally valid study favors adoption, while also avoiding the excessive aggressiveness of $\delta^*$ when that evidence points against adoption. We interpret this as demonstrating the interim credibility of the as-if MMR prescription, thus supporting the rationale for the dynamically consistent minimax regret criterion.


\section{Conclusion} \label{section_conclusion}

This paper studies the dynamic consistency of statistical decision rules justified by ex ante optimality criteria. We show that, for frequentist minimax criteria, dynamically consistent interim preferences can be defined by updating a least favorable prior. This observation yields a diagnostic for assessing whether ex ante minimax rules remain credible after the data realization. In applications to treatment choice and evidence aggregation, this diagnostic reveals that minimax-regret rules can prescribe actions that are difficult to justify at the interim stage  in empirically relevant settings.  To address this issue, we propose and axiomatize a class of dynamically consistent minimax criteria that  provide credible prescriptions in the interim period. Applications to both real and simulated data illustrate their potential value as a tool for guiding statistical decision making.

\appendix

\setcounter{lemma}{0}%
\renewcommand{\thelemma}{A.\arabic{lemma}}
\setcounter{example}{0}%
\renewcommand{\theexample}{A.\arabic{example}}

\setcounter{table}{0}
\renewcommand{\thetable}{A.\arabic{table}}

\section{Proofs} \label{appendix_proof}

\subsection{Proof of Proposition~\ref{prop_posterior_rep}} \label{appendix_proof_LFP}

We have by construction that $\delta^* \in \argmax_{\delta \in \mathcal{D}} \E_{\pi^*}[U(\delta,\theta)]$. By iterated expectations, for any $\delta\in\mathcal D$,
\begin{align*}
\E_{\pi^*}[U(\delta,\theta)]
 &  = 
\int_ Z
\E_{\pi^*(\cdot\mid z)}[u(\delta( z),\theta)]
\, P_{\pi^*}(d  z).
\end{align*}
Suppose, for a contradiction, that there exists some set $B\subseteq Z$ with $P_{\pi^*}(B)>0$ such that $\delta^*( z)$ is not posterior optimal on $B$. For each $ z\in B$, there then exists $a( z)\in\mathcal D( z)$ such that $\E_{\pi^*(\cdot\mid z)}[u(a( z),\theta)]
>
\E_{\pi^*(\cdot\mid z)}[u(\delta^*( z),\theta)].$   Since $\mathcal{D}$ is closed under measurable pasting, we see that the rule $\tilde\delta \in \mathcal D$ that agrees with $\delta^*$ on $B^c$ and satisfies $\tilde\delta( z)=a( z)$ on $B$ attains strictly higher $\pi^*$-expected payoff, in contradiction of the fact that $\delta^*$ best responds to $\pi^*$.   \hfill $\blacksquare$

\subsection{Proof of Proposition~\ref{prop_treatment_LFP}}  \label{Appendix_prop_LFP}

We first prove the following lemma, which gathers and builds on the results from the proofs of Proposition 1(ii) and Corollary 1 in \cite{Stoye_2009_JoE}.

\begin{lemma} \label{lemma_difference}
For the minimax regret rule $\delta^*$  in Fact~\ref{fact_stoye}, regret depends on the state $(\mu_0,\mu_1)$ only through the treatment-effect magnitude $d=|\mu_1-\mu_0|$. In particular,
\begin{align*}
R(\delta^*,(\mu_0,\mu_1) ) = G_N(|\mu_1-\mu_0|),
\end{align*}
where, letting $M$ be the largest odd integer weakly smaller than $N$ and $k \equiv \frac{M-1}{2}$,
\begin{align*}
G_N(d)
\equiv
d
\sum_{j=0}^{k}
\binom{M}{j}
\left(\frac{1+d}{2}\right)^j
\left(\frac{1-d}{2}\right)^{M-j}.
\end{align*}
\end{lemma}

\begin{proof}
Let $A_N = \sum_{j=1}^N \left( \1\{T_j=1,y_j=1\} + \1\{T_j=0,y_j=0\} \right)$, such that $I_N=A_N- \frac{N}{2}$ and $\delta^*$ assigns treatment $1$ iff $A_N>N/2$, assigns treatment $0$ iff $A_N<N/2$, and randomizes iff $A_N=N/2$. Fix a state $\theta = (\mu_0,\mu_1)$, and write $\Delta \equiv \mu_1-\mu_0.$  Under random assignment,  
\begin{align*}
P_{(\mu_0,\mu_1)}
\left(
\{T_j=1,y_j=1\}\cup\{T_j=0,y_j=0\}
\right) = \frac{1}{2}\mu_1+\frac{1}{2}(1-\mu_0) =
\frac{1+\Delta}{2},
\end{align*}
in which case $A_N \sim \mathrm{Binomial}(N, \frac{1+\Delta}{2})$. Denote the corresponding probability mass function  $P_\Delta$. The treatment  assigned by $\delta^*$ depends on $\Theta$ only through $\Delta$.

Now consider the expression for regret. If $\Delta>0$, treatment $1$ is optimal under oracle knowledge of the state. The regret from assigning treatment $0$ is $\Delta$, and the regret from randomizing equally is $\Delta/2$. The same result holds with $|\Delta|$ when $\Delta < 0$ by symmetry, and the regret equals 0 when $\Delta = 0$. Since $d \equiv |\Delta|$, we have 
\begin{align*}
R(\delta^*,(\mu_0,\mu_1))
=
d
\left[
P_{\Delta}\left(A_N<\frac{N}{2}\right)
+
\frac{1}{2}
P_{\Delta}\left(A_N=\frac{N}{2}\right)
\right].
\end{align*}

It remains to write this expression in terms of $G_N(\cdot)$. Let $p \equiv \frac{1+d}{2}$, and first consider the case where $N$ is odd, such that $M = N$. Define $k \equiv \frac{N-1}{2}$, such that  $\{A_N<N/2\} =  \{A_N\leq k \}$ and $\{A_N = \frac{N}{2}\} = \emptyset$. It follows that
\begin{align*}
R(\delta^*,(\mu_0,\mu_1)) =  
d
\sum_{j=0}^{k}
\binom{N}{j}
p^j(1-p)^{N-j} = G_N(d).
\end{align*}

Next, if $N$ is even, define $m \equiv \frac{N}{2}$, such that 
\begin{align*}
R(\delta^*,(\mu_0,\mu_1)) = d
\left[
P_\Delta(A_N<m)
+
\frac{1}{2}P_\Delta(A_N=m)
\right].
\end{align*}
Let $C_{N-1}\sim\mathrm{Binomial}(N-1,p)$, independent of a final Bernoulli draw $B \sim \Bern(p)$. We can then write
\begin{align*}
P_\Delta(A_N<m)
&=
P(C_{N-1}\leq m-2)
+
(1-p)P(C_{N-1}=m-1) \\ 
P_\Delta(A_N=m)
& =
pP(C_{N-1}=m-1)
+
(1-p)P(C_{N-1}=m).
\end{align*}
Since $N-1=2m-1$, we have 
\begin{align*}
    (1-p)P(C_{N-1}=m) &=  \binom{2m-1}{m} p^m (1-p)^m \\
    & = \binom{2m-1}{m-1} p^m (1-p)^m   = pP(C_{N-1}=m-1).
\end{align*}
Substituting yields $P(A_N<m)
+
\frac{1}{2}P(A_N=m)
=
P(C_{N-1}\leq m-1)$, and 
\begin{align*}
R(\delta^*,(\mu_0,\mu_1))
=
d
\sum_{j=0}^{m-1}
\binom{N-1}{j}
p^j(1-p)^{N-1-j}.
\end{align*}
This expression equals $G_N(d)$, with $M=N-1$ and $k=(M-1)/2=m-1$. 
\end{proof}

Recall that $\theta = (\mu_0, \mu_1) \in [0,1]^2$, in which case $|\mu_1 - \mu_0| \in [0,1]$ for all $\theta \in \Theta$. The following lemma shows that  $G_N(\cdot)$ has a unique maximum over the collection of possible $|\mu_1 - \mu_0|$.

\begin{lemma}
The function $G_N(\cdot)$ has a unique maximizer on $[0,1]$.
\end{lemma}

\begin{proof}
Let $p \equiv \frac{1+d}{2}$ and $r \equiv p- \frac{1}{2} = \frac{d}{2}$, such that
\begin{align*}
    G_N(d) & = 2r \sum_{j=0}^k \binom{2k+1}{j} p^j (1-p)^{2k+1 - j} \\ 
    & = 2r \cdot \frac{(2k+1)!}{k!k!}
\int_0^{1-p} t^k(1-t)^k  \, dt \\
& = 2r \cdot \frac{(2k+1)!}{k!k!}
\int_r^{1/2}
\left(\frac{1}{4}-u^2\right)^k  du \tag{$t = \frac{1}{2} + u$} \\
& \propto
\underbrace{r \int_r^{1/2} \left(\frac{1}{4}- u^2\right)^k du}_{\equiv H(r) },
\end{align*}
where the second step follows from the  fact that for $X \sim \text{Binomial}(n,p)$, $\pr(X \leq k) = (n-k) \binom{n}{k} \int_0^{1-p} t^{n-k-1} (1-t)^k \, dt$.  Hence, we see that  maximizing $G_N(d)$ over $d \in[0,1]$ is equivalent to maximizing $H(r)$ over $r\in[0,\frac{1}{2}]$. 

If $k=0$, then $H(r)=r( \frac{1}{2} -r),$  which has a unique maximizer, as desired. Now suppose $k\geq 1$. Define $g(r) \equiv  (\frac{1}{4}-r^2 )^k$ and  $ S(r) \equiv  \int_r^{1/2}g(u) \, du$, such that  $H(r)=rS(r)$ and  $H'(r)=S(r)-rg(r)$.  Differentiating once more, we obtain
\begin{align*}
H''(r) &  = -2g(r)-rg'(r) =  2 \left(\frac{1}{4}-r^2\right)^{k-1} \left((k+1)r^2-\frac{1}{4}\right). 
\end{align*}

Since $r \in [0,\frac{1}{2}]$, the  first term is positive, and the second term  determines the sign of $H''(r)$.  We then see that $H'$ is strictly decreasing on $(0,\frac{1}{2\sqrt{k+1}})$ and  strictly increasing on $(\frac{1}{2\sqrt{k+1}},\frac{1}{2})$. Moreover, $H'(0)>0$, $H(0)=H(\frac{1}{2}) = 0$, $H(r) > 0$ for $r \in (0,\frac{1}{2})$, and  $H'(r)<0$ for $r$ sufficiently close to $1/2$. Since $H'$ is strictly increasing near the right endpoint and converges to $0$ there from below, it remains negative on the second interval. In other words, $H'$ crosses zero exactly once, in which case we have that $H$  has a unique maximizer, as desired.
\end{proof}

Let $d_N$ denote the unique maximizer of $G_N$. We know from Lemma~\ref{lemma_difference} that $R(\delta^*,(\mu_0,\mu_1)) = G_N(|\mu_1-\mu_0|)$. Since  every least favorable prior $\pi^*$ places probability one on states that maximize the regret of $\delta^*$, it follows that  $|\mu_1-\mu_0|=d_N$ $\pi^*\text{-almost surely}$, as desired.\footnote{Recall that in the zero-sum game with Nature, any least favorable $\pi'$ that best-responds to the minimax rule $\delta'$ also best responds to any other minimax rule $\delta^*$.} \hfill $\blacksquare$  \\

We conclude by deriving  the maximizer $d_N$ in  the example in Section \ref{section_treatment_choice}.

\begin{example} \label{ex_difference_treatment}
    Suppose $N=1100$, such that  $M = 1099$, $k = 549$, and
    \begin{align*}
        G_{1100}(d) & = d \sum_{j=0}^{549} \binom{1099}{j} \left( \frac{1 + d}{2} \right)^j \left( \frac{1-d}{2} \right)^{1099- j}.
    \end{align*}
    Numerical optimization yields $d_{1100} \equiv  \argmax_{d \in [0,1]} G_{1100}(d) \approx 0.023$.  
\end{example}

\subsection{Proof of Proposition~\ref{prop_Yata_LFP}} \label{Appendix_prop_Yata}

Since the payoff from policy $0$ is normalized to zero and the payoff from policy $1$ is $\theta_3$, the oracle chooses policy $1$ when $\theta_3\geq 0$ and policy $0$ when $\theta_3<0$.

First consider states with $\theta_3\geq 0$. Since $ z_2\sim N(\theta_2,\sigma^2)$ and $\delta^*( z)=\1 \{ z_2\geq 0\}$,
\begin{align*}
R(\delta^*,\theta)
=
\theta_3
P_\theta( z_2<0)
=
\theta_3
\Phi\left(-\frac{\theta_2}{\sigma}\right).
\end{align*}
For fixed $\theta_2$, this expression is increasing in $\theta_3$. The constraint $|\theta_2-\theta_3|\leq C_2$ over $\Theta$ then implies that regret is maximized by setting $\theta_3=\theta_2+C_2$, conditional on $\theta_2$. Choosing the state $\theta$ that best responds to $\delta^*$ then reduces to maximizing
\begin{align*}
h(s)
=
(s+C_2)\Phi\left(-\frac{s}{\sigma}\right)
\end{align*}
over values $s=\theta_2$ for which $s+C_2\geq 0$. 

Observe that for  $s\in[-C_2,0]$,
\begin{align*}
h'(s) & =
\Phi\left(-\frac{s}{\sigma}\right)
-
\frac{s+C_2}{\sigma}
\phi\left(\frac{s}{\sigma}\right)  \geq
\frac{1}{2}
-
\frac{C_2}{\sigma}\phi(0),
\end{align*}
which is strictly positive under our assumption that $2\phi(0)C_2<\sigma$. The optimal $s$ is thus non-negative,  in which case the definition of $\epsilon^*$ allows us to conclude that  $\theta_2=\epsilon^*$ and $ 
\theta_3=\epsilon^* +C_2.$ Regret does not depend on $\theta_1$, so we see that the collection of regret-maximizing states is $\left\{  \theta \in \Theta : \theta_2=\epsilon^*, \,\,  \theta_3=\epsilon^* +C_2 \right\} = \Theta^+.$  

In the case where $\theta_3 < 0$, we have $R(\delta^*,\theta) = -\theta_3 P_\theta( z_2\geq 0) = -\theta_3 \Phi\left(\frac{\theta_2}{\sigma}\right),$  and a symmetric argument then reveals that the collection of regret-maximizing states is $\Theta^-$.  All together, it follows that any strategy $\pi^*$ (i.e., least favorable prior) chosen by Nature in the zero-sum game against the DM must be supported on $\Theta^+ \cup \Theta^-$, as desired. \hfill $\blacksquare$

\subsection{Proof of Theorem \ref{thm_rep_loss}}  \label{appendix_proof_loss}

Necessity of the axioms is routine; we only prove sufficiency here.

    By $\Theta$-MEU, the induced preference $\succsim_z$ over $\Theta$-measurable acts is represented by $\mathcal{L}_z(\delta ) = \min_{q_z \in \mathcal{Q}_z} \sum_{\theta \in \Theta} u(\delta(z), \theta) q_z(\theta)$ for unique closed and convex $\mathcal{Q}_z \subseteq \Delta(\Theta)$.  Let $\delta_z \in \R^\Theta$ denote the $\Theta$-measurable act induced by any $\delta$ at $z$. Observe that $\delta_z \sim_z \mathcal{L}_z(\delta_z)$, where $\mathcal{L}_z(\delta_z)$ is the constant act that maps to $\mathcal{L}_z(\delta)$. 

    Let $Z = \{z_1,...,z_{|Z|}\}$ without loss. By definition of the induced preference $\succsim_{z_1}$,  
    \begin{align*}
        \delta_{z_1} \sim_{z_1} \mathcal{L}_{z_1}(\delta_{z_1}) \hspace{10pt} \Rightarrow \hspace{10pt} \delta  \sim  ( \mathcal{L}_{z_1}(\delta_{z_1}) )_{ \{z_1\} } \delta.
    \end{align*}
    Likewise, taking $\phi = ( \mathcal{L}_{z_1}(\delta_{z_1}) )_{ \{z_1\} } \delta$ in the definition of $\succsim_{z_2}$,  we have
    \begin{align*}
        \delta_{z_2} \sim_{z_2} \mathcal{L}_{z_2}(\delta_{z_2}) \hspace{10pt} \Rightarrow \hspace{10pt}  ( \mathcal{L}_{z_1}(\delta_{z_1}) )_{ \{z_1\} } \delta  \sim  ( \mathcal{L}_{z_1}(\delta_{z_1}) )_{ \{z_1\} } (\mathcal{L}_{z_2}(\delta_{z_2}))_{ \{z_2\} }  \delta
    \end{align*}
    Repeating inductively, we can construct the $Z$-measurable act $\bar \delta$  that satisfies 
    \begin{align*}
          \delta \sim (\mathcal{L}_{z_1}(\delta_{z_1}))_{ \{z_1\} } \delta \sim  (\mathcal{L}_{z_1}(\delta_{z_1}))_{ \{z_1\} } (\mathcal{L}_{z_2}(\delta_{z_2}))_{ \{z_2\} }  \delta \sim \hdots \sim   \bar \delta.
    \end{align*}

    By $Z$-SEU, we know that $\succsim$ restricted to $Z$-measurable acts is represented by $\mathcal{L}(\delta) = \sum_{z \in Z} \delta(z,\theta) \mu(z)$ for some unique $\mu \in \Delta(Z)$. For any $\delta \in \mathcal{D}$, the above construction identifies a $Z$-measurable equivalent $\bar \delta$. Thus, $\succsim$ is represented by
    \begin{align*}
        U(\delta)  = \mathcal{L}(\bar \delta) & = \sum_{z \in Z} \mu(z) \bar \delta(z,\theta)  \\ 
        & = \sum_{z \in Z}  \mu(z)  \min_{q_z \in \mathcal{Q}_z} \sum_{\theta \in \Theta} u(\delta(z), \theta) q_z(\theta)  = \E_\mu\left[ \min_{q \in \mathcal{Q}_z } \E_q\left[ u(\delta(z), \theta) \right]\right] ,
    \end{align*}
    as desired.  \hfill $\blacksquare$

\subsection{Proof of Theorem~\ref{thm_rep_regret}} \label{appendix_proof_regret} 

Necessity of the axioms is routine; we only prove sufficiency here.

   Suppose that $C(\cdot)$ satisfies Axioms 4-6.  For any finite menu $M$ and $\delta\in M$, define the regret act
\begin{align*}
     r_\delta^M(z,\theta)
    \equiv
    \max_{\beta\in M}\beta(z,\theta)-\delta(z,\theta).
\end{align*}
Let $\mathbb R^{Z\times\Theta}_+$ denote the set of non-negative regret acts over $Z\times\Theta$.

We use the following two results from \cite{Stoye_2011_JET}, restated using our notation. Because IIA on $Z$-measurable acts imply IIA over constant acts, the premise of both results are satisfied under Axioms~\ref{axiom_stoye}-\ref{axiom_sep_choice}.

\begin{lemma}[Lemma 2 of \cite{Stoye_2011_JET}]
    If  $C(\cdot)$ satisfies non-triviality, monotonicity, independence, INA, mixture continuity, and IIA for constant acts, then there exists a functional $I:\mathbb R^{Z \times \Theta}_+\to \mathbb R_+$ 
such that 
\begin{align*}
     C(M)
    =
    \arg\min_{\delta\in M} I(r_\delta^M).
\end{align*}
\end{lemma}

\begin{theorem}[Theorem 4 of \cite{Stoye_2011_JET}]
     If  $C(\cdot)$ satisfies Axiom~\ref{axiom_stoye} and IIA for constant acts, then there exists a unique compact
convex $\Gamma\subseteq \Delta(Z \times \Theta)$ such that
\begin{align}
     I(r)
    =
    \max_{\gamma\in\Gamma}\sum_{(z,\theta) \in Z \times \Theta}r(z,\theta)\gamma(z,\theta), \label{eq_stoye_I}
\end{align}
and therefore
\begin{align}
    C(M)
    =
    \arg\min_{\delta\in M}
    \max_{\gamma\in\Gamma}
    \sum_{(z,\theta) \in Z \times \Theta}
    \left(
        \max_{\beta\in M}\beta(z,\theta)-\delta(z,\theta)
    \right)\gamma(z,\theta). \label{eq_stoye_rep}
\end{align}
\end{theorem}

Define the induced preference $\preceq_R$ over regret acts by taking $r\preceq_R r'$ if $I(r)\le I(r')$ and
$r\sim_R r'$ if $I(r)=I(r')$. Let  $\Gamma_Z
    \equiv
    \left \{\gamma_Z(\cdot) = \sum_{\theta} \gamma(\cdot,\theta) :\gamma\in\Gamma \right\}
    \subseteq \Delta(Z)$  denote the collection of $Z$-marginals induced by $\Gamma$. 

First, consider the restriction of $C$ to menus of $Z$-measurable utility acts. For such menus, write $\delta(z,\theta)=\widehat\delta(z)$. The
representation in \eqref{eq_stoye_rep} then reduces to
\begin{align}
    C(M)
    =
    \arg\min_{\delta\in M}
    \max_{\eta\in\Gamma_Z}
    \sum_{z\in Z}
    \left(
        \max_{\beta\in M}\widehat\beta(z)
        -
        \widehat\delta(z)
    \right) \eta(z).
    \label{eq_rep_Z}
\end{align}
By Axiom 5, this restricted choice correspondence satisfies IIA on $Z$-measurable acts. Together with the monotonicity, independence, and mixture continuity conditions inherited from Axiom 4, the usual finite Anscombe-Aumann argument recovers a unique SEU measure $\mu\in\Delta(Z)$, such that
choice over  $Z$-measurable acts is represented by $\widehat\delta \mapsto
    \sum_{z\in Z}\hat \delta(z) \mu(z)$.  We can thus rewrite \eqref{eq_rep_Z} as 
\begin{align*}
    C(M)
    =
    \arg\min_{\delta\in M}
    \sum_{z\in Z}
    \left(
        \max_{\beta\in M}\widehat\beta(z)
        -
        \widehat\delta(z)
    \right)\mu(z). 
\end{align*}

By uniqueness of \cites{Stoye_2011_JET} endogenous prior minimax regret representation applied to the restricted $Z$-measurable problem, it follows that $  \Gamma_Z=\{\mu\}.$  That is, every $\gamma\in\Gamma$ has the $Z$-marginal $\mu$, such that
\begin{align}
    I(r)
    =
    \sum_{z\in Z}\mu(z)\widehat r(z). \label{eq_our_I}
\end{align}
for every non-negative $Z$-measurable regret act $r(z,\theta)=\widehat r(z)$.

For each $z$ with $\mu(z)>0$, define
\begin{align*}
    \mathcal{Q}_z
    \equiv
    \left\{
        q_z\in\Delta(\Theta):
        \exists \gamma\in\Gamma
        \text{ such that }
        q_z(\theta)=\frac{\gamma(z,\theta)}{\mu(z)}
        \text{ for every }\theta\in\Theta
    \right\}. 
\end{align*}
Since $\Gamma$ is compact and convex and the map
$\gamma\mapsto \gamma(\cdot\mid z)$ is continuous and affine over the set
of priors with $Z$-marginal $\mu$, each $\mathcal{Q}_z$ is compact and
convex.\footnote{If $\mu(z)=0$, the choice of $\mathcal{Q}_z$ is irrelevant for the
ex ante representation. When discussing uniqueness properties of the representation, we assume the $\mu$ has full support; this can  easily be motivated axiomatically.}

For any non-negative $\Theta$-measurable regret act $a:\Theta\to\mathbb R_+$, define
\begin{align*}
    v_z(a)
    \equiv
    \max_{q_z\in \mathcal{Q}_z}
    \sum_{\theta\in\Theta} a(\theta)q_z(\theta). 
\end{align*}
Let $\mathbf 0$ denote the zero regret act on $Z\times\Theta$. The definition of $\mathcal{Q}_z$ and \eqref{eq_stoye_I} implies
\begin{align*}
    I(a_{\{z\}}\mathbf 0) & = \max_{\gamma \in \Gamma} \sum_{\theta \in \Theta} a(\theta) \gamma(z,\theta)  =\max_{q_z \in \mathcal{Q}_z} \sum_{\theta \in \Theta} a(\theta) \mu(z) q_z(\theta)  =  \mu(z)v_z(a). 
\end{align*}
Since we also have $ I(v_z(a)_{\{z\}}\mathbf 0)  =  \mu(z)v_z(a),$ it follows that 
\begin{align}
    a_{\{z\}}\mathbf 0
    \sim_R
    v_z(a)_{\{z\}}\mathbf 0. \label{eq_indiff}
\end{align}

Next, we show that $Z$-separable choice (Axiom~\ref{axiom_sep_choice}) implies that  $\preceq_R$ satisfies the $Z$-separability axiom from the setting of minimax loss (Axiom~\ref{axiom_z_sep_loss}). 

\begin{lemma}
The induced regret relation $\preceq_R$ satisfies
\begin{align*}
    a_{\{z\}}h \preceq_R b_{\{z\}}h
    \quad\Longleftrightarrow\quad
    a_{\{z\}}k \preceq_R b_{\{z\}}k 
\end{align*}
for every $z \in Z$,  $a,b : \theta \to \R_+$, and $h,k : Z \times \Theta \to \R_+$.
\end{lemma} 
\begin{proof}
Consider the regret act defined by
\begin{align*}
    \mathbf 1_{-z}(z',\theta)
    \equiv
    \begin{cases}
        0, & z'=z,\\
        1, & z'\ne z.
    \end{cases}
\end{align*}
Since every $\gamma\in\Gamma$ has $Z$-marginal $\mu$, we have 
\begin{align}
    I(r+\lambda \mathbf 1_{-z})
    &=
    \max_{\gamma\in\Gamma}
    \sum_{(z',\theta)\in Z\times\Theta}
    \left(r(z',\theta)+\lambda\mathbf 1_{-z}(z',\theta)\right)
    \gamma(z',\theta) \nonumber \\
    &=
    I(r)+\lambda(1-\mu(z)) \label{eq_lambda}
\end{align}
for every $r\in\mathbb R^{Z\times\Theta}_+$ and $\lambda\geq 0$. 


Choose $\lambda_h,\lambda_k\geq 0$, such that $ I(a_{\{z\}}(h+\lambda_h\mathbf 1_{-z}))
    =
    I(a_{\{z\}}(k+\lambda_k\mathbf 1_{-z}))$.\footnote{ If $ I(a_{\{z\}}h)\leq I(a_{\{z\}}k)$, this can be done by defining $\lambda_h
    =
    \frac{I(a_{\{z\}}k)-I(a_{\{z\}}h)}{1-\mu(z)}$, $
    \lambda_k=0$,  and likewise if the inequality holds in the opposite direction.}  Define $ \bar h\equiv h+\lambda_h\mathbf 1_{-z},$ $
    \bar k\equiv k+\lambda_k\mathbf 1_{-z},$  
and set $m \equiv
    I(a_{\{z\}}\bar h)
    =
    I(a_{\{z\}}\bar k).$  Note that $a_{ \{z\} } h + \lambda_h \mathbf{1}_{-z} = a_{ \{z\} } (h + \lambda_h \mathbf{1}_{-z}) = a_{ \{z\} } \bar h$, in which case we have
\begin{align}
    a_{\{z\}}h \preceq_R b_{\{z\}}h
   & \spiff I(a_{\{z\}}  h) \leq I(b_{\{z\}} h )  \nonumber \\
   & \spiff I(a_{\{z\}} h + \lambda_h \mathbf{1}_{-z} ) \leq I(b_{\{z\}} h+ \lambda_h \mathbf{1}_{-z} ) \tag{by \eqref{eq_lambda}}   \\
   & \spiff
    a_{\{z\}}\bar h \preceq_R b_{\{z\}}\bar h.
  \label{eq_bar_equiv1} 
\end{align}
The same argument with $h$ replaced with $k$ yields
\begin{align}
    a_{\{z\}}k \preceq_R b_{\{z\}}k
    & \spiff 
    a_{\{z\}}\bar k \preceq_R b_{\{z\}}\bar k. \label{eq_bar_equiv2}
\end{align}

We now construct a menu on which Axiom~\ref{axiom_sep_choice} applies. Define
\begin{align*}
    f_h \equiv -a_{\{z\}}\bar h,
    & \qquad
    f_k \equiv -a_{\{z\}}\bar k, \\ 
    g_h \equiv -b_{\{z\}}\bar h,
    & \qquad
    g_k \equiv -b_{\{z\}}\bar k.
\end{align*}
For each state  $s=(z',\theta)\in Z\times\Theta$, let $e^s_K$ be the utility act
\begin{align*}
    e^s_K(t)
    =
    \begin{cases}
        0, & t=s,\\
        -K, & t\ne s. 
    \end{cases}
\end{align*}
Consider the menu $ M_K
    \equiv
    \{f_h,f_k\}\cup \{e^s_K:s\in Z\times\Theta\}$, such that the statewise maximum utility in $M_K$ is zero. Hence, the regret act corresponding to $f_h$ and $f_k$ are $r^{M_K}_{f_h} = a_{\{z\}}\bar h$ and $r^{M_K}_{f_k} = a_{\{z\}}\bar k$, respectively,  such that
    \begin{align}
        I(r^{M_K}_{f_h}) = I(r^{M_K}_{f_k}) = m.  \label{eq_regret_eq}
    \end{align}

    For the normalizing act $e^s_K$, its regret is zero at $s$ and $K$ at all other states. If $s=(z_s,\theta_s)$, then every $\gamma\in\Gamma$
assigns probability $1-\mu(z_s)$ to states with data realization different  from $z_s$. Thus, we have that
    \begin{align*}
        I(r_{e^s_K}^{M_K} ) = \max_{\gamma \in \Gamma} \sum_{(z,\theta) \neq s}  K \gamma(z,\theta) \geq K(1-\mu(z_s)).
    \end{align*}
Choose $K$ large enough such that 
\begin{align*}
   I(r_{e^s_K}^{M_K} )  \geq K\min_{z'\in Z}(1-\mu(z'))>m
\end{align*}
for all $s$. In other words, the normalizing act $e_K^s$ is  never chosen from the menu $M_K$, and \eqref{eq_regret_eq} implies $ f_h,f_k\in C(M_K)$.   

Now add $g_h$ to $M_K$. Since $g_h\le 0$ statewise, the statewise utility frontier remains zero. The regret act corresponding to $g_h$ in $M_K\cup\{g_h\}$ is $b_{\{z\}}\bar h$, while the regrets of all acts already in $M_K$ are unchanged. Therefore,
\begin{align*}
    f_h\in C(M_K\cup\{g_h\})
    & \spiff I( r_{f_h}^{M_k \cup \{g_h\}}) \leq I( r_{g_h}^{M_k \cup \{g_h\}})  \\
    & \spiff  I(a_{\{z\}}\bar h)\le I(b_{\{z\}}\bar h), 
\end{align*}
and the same argument for $f_k, g_k$ yields 
\begin{align*}
    f_k\in C(M_K\cup\{g_k\})
   \spiff 
    I(a_{\{z\}}\bar k)\le I(b_{\{z\}}\bar k). 
\end{align*}

Because $f_h,f_k\in C(M_K)$, Axiom~\ref{axiom_sep_choice}  then implies
\begin{align*}
    f_h\in C(M_K\cup\{g_h\})
    \spiff
    f_k\in C(M_K\cup\{g_k\}),
\end{align*}
which is equivalent to  $I(a_{\{z\}}\bar h)\le I(b_{\{z\}}\bar h)
    \iff
    I(a_{\{z\}}\bar k)\le I(b_{\{z\}}\bar k).$  It then follows from \eqref{eq_bar_equiv1} and \eqref{eq_bar_equiv2} that
\begin{align*}
    I(a_{\{z\}}h)\le I(b_{\{z\}}h)
    \iff 
    I(a_{\{z\}}k)\le I(b_{\{z\}}k),
\end{align*}
in which case  $ a_{\{z\}}h \preceq_R b_{\{z\}}h \iff
    a_{\{z\}}k \preceq_R b_{\{z\}}k,$  
as desired.
\end{proof}

Applying  the previous lemma to \eqref{eq_indiff} yields 
\begin{align*}
   a_{\{z\}}h
    \sim_R
    v_z(a)_{\{z\}}h
\end{align*}
for any regret act $h$. For every regret act $r$, let $r_z$ be the $\Theta$-measurable act defined by $r_z(\theta) = r(z,\theta)$. We can then proceed inductively  as in the proof of Theorem~\ref{thm_rep_loss} to recover a $Z$-measurable regret act $\overline r \sim_R r$ that satisfies 
\begin{align*}
    \overline r(z,\theta) 
    = v_z(r_z) = 
    \max_{q_z\in \mathcal{Q}_z}
    \sum_{\theta'\in\Theta}
    r(z,\theta')q_z(\theta')
\end{align*}
for all $(z,\theta) \in Z \times \Theta$. Plugging into  \eqref{eq_our_I} then yields
\begin{align*}
    I(r)
    &=
    I(\overline r)  =
    \sum_{z\in Z}\mu(z)
    \max_{q_z\in \mathcal{Q}_z}
    \sum_{\theta\in\Theta}
    r(z,\theta)q_z(\theta).
\end{align*}

Finally,  substituting $r=r_\delta^M$ reveals that for any finite menu $M$, we have
\begin{align*}
    C(M)
    &=
    \arg\min_{\delta\in M} I(r_\delta^M) \\
    &=
    \arg\min_{\delta\in M}
    \sum_{z\in Z}\mu(z)
    \max_{q_z\in \mathcal{Q}_z}
    \sum_{\theta\in\Theta}
    r_\delta^M(z,\theta)q_z(\theta) \\
    &=
    \arg\min_{\delta\in M}
    \sum_{z\in Z}\mu(z)
    \max_{q_z\in \mathcal{Q}_z}
    \sum_{\theta\in\Theta}
    \left(
        \max_{\beta\in M}\beta(z,\theta)
        -
        \delta(z,\theta)
    \right)
    q_z(\theta),
\end{align*}
which coincides with the dynamically consistent minimax regret representation.

It remains to discuss uniqueness. The marginal $\mu$ is uniquely
identified by the restriction of choice to $Z$-measurable utility acts. For every $z$ with $\mu(z)>0$, the support function $a \mapsto v_z(a) =    \frac{1}{\mu(z)}I(a_{\{z\}}\mathbf 0)$  is unique. Since a compact convex subset of  $\Delta(\Theta)$ is uniquely determined by its support
function, $\mathcal{Q}_z$ is uniquely identified.  The full support assumption on $\mu$ then yields the desired result.\footnote{Without full support,
$\mathcal{Q}_z$ is not identified at $\mu$-null data realizations and these sets are irrelevant for the ex ante representation.}  \hfill $\blacksquare$

\subsection{Proof of Lemma~\ref{lem:asif-Ta}} \label{appendix_lem_asif_Ta}

The as-if rule is minimax regret over the confidence interval $C=[\hat\tau-c\,\widehat{\mathrm{se}},\ \hat\tau+c\,\widehat{\mathrm{se}}]$ for the effect $\tau$, where $c>0$ is the critical value of the chosen confidence level and $\widehat{\mathrm{se}}>0$. For the binary decision, the regret of providing at effect $\tau$ is $\max\{0,-\tau\}$ and the regret of withholding is $\max\{0,\tau\}$, so the worst-case regrets over $C$ are
\begin{align*}
\max_{\tau\in C}\max\{0,-\tau\}=\max\{0,\,c\,\widehat{\mathrm{se}}-\hat\tau\},\qquad
\max_{\tau\in C}\max\{0,\tau\}=\max\{0,\,\hat\tau+c\,\widehat{\mathrm{se}}\}.
\end{align*}

The as-if rule provides when the first term is smaller. If $\hat\tau>0$, the second term equals $\hat\tau+c\,\widehat{\mathrm{se}}>0$ and strictly exceeds the first, since $c\,\widehat{\mathrm{se}}-\hat\tau<\hat\tau+c\,\widehat{\mathrm{se}}$; the as-if rule provides. If $\hat\tau<0$ the comparison reverses and it withholds, and if $\hat\tau=0$ the two are equal. Hence the as-if rule provides if and only if $\hat\tau>0$, for every $c>0$; in particular its action does not depend on the confidence level used to form $C$.

The test $T_a$ provides if and only if $z>z_\alpha$, where $z=\hat\tau/\widehat{\mathrm{se}}$ and $z_\alpha=\Phi^{-1}(1-\alpha)$. For $\alpha\in(0,\tfrac12)$ we have $z_\alpha>0$, so $T_a$ providing implies $\hat\tau/\widehat{\mathrm{se}}>z_\alpha>0$, and therefore $\hat\tau>0$ because $\widehat{\mathrm{se}}>0$. By the previous paragraph the as-if rule then provides. On the event $\{\delta^\ast=\mathrm{hold},\,T_a=\mathrm{provide}\}$ both rules provide, and the as-if action is the same at every confidence level. \hfill $\blacksquare$

\section{Details of Empirical Application} \label{appendix_empirics}

\setcounter{table}{0}
\renewcommand{\thetable}{B.\arabic{table}}

\subsection{Covariate Cells}
\label{app:protocol} \label{app:cells}


We form cells from the twelve pre-treatment characteristics that organize the
experiment: the men's elicited second-order beliefs about how many other men support women working
(in general, in semi-segregated environments, and on a minimum wage), the misperception (wedge)
between each belief and the truth, the confidence attached to it, and education, current
employment, and the wife's employment. These belief, wedge, and confidence measures are elicited
before the information is provided; the post-treatment labor-demand guess is excluded, since it
follows the sign-up decision and is treated as an outcome in the replication code. Taking 
covariates and all two- and three-way intersections,  and keeping cells with at least five men in
each arm, the minimax regret rule withholds in $4{,}273$ cells.

The cells we study do not form a single partition of the sample. We enumerate subgroups defined by one, two, and three of the pre-treatment characteristics, so a one-characteristic cell (e.g., low education) strictly contains its finer refinements, and cells built from different characteristics cross-cut one another. What licenses the exercise is not that the cells partition the data, but that, because treatment is assigned by an unconditional fair coin, \emph{every} such subgroup is itself a valid randomized experiment; within each cell the finite-sample minimax-regret problem of \citet{Stoye_2009_JoE} is well posed, and the comparison of $\delta^\ast$ with the hypothesis-testing and as-if MMR rules is a legitimate instance of the interim problem of Section~\ref{section_treatment_choice}.

\subsection{Pairwise Disagreements}
\label{app:disagreement}

Tables~\ref{tab:disagreement-05} and~\ref{tab:disagreement-10} report, across the $18{,}291$ covariate cells of Appendix~\ref{app:protocol}, the number of cells on which each pair of rules---the minimax-regret rule $\delta^\ast$, the hypothesis-testing rule $T_\alpha$, and the as-if MMR rule---prescribes a different interim action, at $\alpha=0.05$ and $\alpha=0.10$, respectively. Each entry splits the disagreement into its two directions: the top number counts cells where the row rule provides while the column rule withholds, and the bottom number cells where the row withholds while the column provides. The lower-right entry is zero at both levels, as predicted by Lemma~\ref{lem:asif-Ta}. Cells on which a rule is exactly indifferent ($I_N=0$ for $\delta^\ast$, or $\hat\tau=0$ for the as-if rule) are excluded. In Table~\ref{tab:disagreement-05} the as-if MMR rule is formed from a $95\%$ confidence interval to match $\alpha=0.05$. Because the as-if MMR action depends only on the sign of $\hat\tau$, the as-if entries remain the same in Table~\ref{tab:disagreement-10}, and the two tables differ only through the level of $T_\alpha$.

Because the covariate cells overlap, the counts in Table~\ref{tab:bursztyn-ta} and Tables~\ref{tab:disagreement-05}--\ref{tab:disagreement-10} correspond to  different objects. The pairwise-disagreement tables of this section count \emph{cells}, whereas Table~\ref{tab:bursztyn-ta} lists \emph{distinct subgroups} (i.e., covariate combinations that select samples with identical arm sizes and success counts are shown once). Hence the event $\{\delta^\ast=\text{withhold},\,T_\alpha=\text{provide}\}$ occurs in  $21$ cells but $18$ distinct subgroups at $\alpha=0.10$, and $3$ cells but $2$ distinct subgroups at $\alpha=0.05$.



\begin{table}[ht]
\centering
\caption{Pairwise disagreements at $\alpha=0.05$. Top: row provides, column withholds. Bottom: row withholds, column provides (\% of $18{,}291$ cells in parentheses).}
\label{tab:disagreement-05}
\newcommand{\twocell}[2]{%
  \begin{tabular}[c]{@{}c@{}}
    #1\\
    #2
  \end{tabular}%
}
\renewcommand{\arraystretch}{1}
\begin{tabular}{l ccc}
\toprule
 & $\delta^\ast$ (minimax) & $T_\alpha$ (test) & as-if MMR \\
\midrule
$\delta^\ast$ (minimax) & --- & & \\[0.4em]
$T_\alpha$ (test)
  & \shortstack{$3$ {\footnotesize($0.0\%$)}\\ $10{,}304$ {\footnotesize($56.3\%$)}}
  & --- & \\[0.4em]
as-if MMR
  & \shortstack{$994$ {\footnotesize($5.4\%$)}\\ $552$ {\footnotesize($3.0\%$)}}
  & \shortstack{$11{,}349$ {\footnotesize($62.0\%$)}\\ $0$ {\footnotesize($0.0\%$)}}
  & --- \\
\bottomrule
\end{tabular}
\end{table}

\begin{table}[ht]
\centering
\caption{Pairwise disagreements at $\alpha=0.10$. Top: row provides, column withholds. Bottom: row withholds, column provides (\% of $18{,}291$ cells in parentheses).}
\label{tab:disagreement-10}
\renewcommand{\arraystretch}{1}
\begin{tabular}{l ccc}
\toprule
 & $\delta^\ast$ (minimax) & $T_\alpha$ (test) & as-if MMR \\
\midrule
$\delta^\ast$ (minimax) & --- & & \\[0.4em]
$T_\alpha$ (test)
  & \shortstack{$21$ {\footnotesize($0.1\%$)}\\ $8{,}434$ {\footnotesize($46.1\%$)}}
  & --- & \\[0.4em]
as-if MMR
  & \shortstack{$994$ {\footnotesize($5.4\%$)}\\ $552$ {\footnotesize($3.0\%$)}}
  & \shortstack{$9{,}453$ {\footnotesize($51.7\%$)}\\ $0$ {\footnotesize($0.0\%$)}}
  & --- \\
\bottomrule
\end{tabular}
\end{table}

\bibliography{ref}

\newpage

\end{document}